\RequirePackage{fix-cm}

\documentclass{amsart}
\usepackage{amsmath,amssymb,amsfonts,amsthm}
\usepackage{amsaddr}
\theoremstyle{plain}
\newtheorem{theorem}{Theorem}
\newtheorem{corollary}[theorem]{Corollary}
\newtheorem{lemma}[theorem]{Lemma}
\newtheorem{proposition}[theorem]{Proposition}
\theoremstyle{definition}
\newtheorem{definition}[theorem]{Definition}

\theoremstyle{remark}
\newtheorem{remark}[theorem]{Remark}

\usepackage{latexsym}

\usepackage{amsmath,xspace}
\usepackage{etex}
\usepackage{centernot}
\usepackage{stmaryrd}
\usepackage{wrapfig}
\usepackage{amsfonts}
\usepackage{mathtools}
\usepackage{color}
\definecolor{lgray}{gray}{0.9}
\usepackage{fixltx2e}
\usepackage{fancybox}
\usepackage{ifthen}
\usepackage{pgf}
\usepackage{tikz}
\usetikzlibrary{arrows,automata}
\usepackage{graphics} 
\usepackage{subfig}
\usepackage{proof}
\usepackage{hyperref}
\usepackage{soul}
\usepackage{cancel}
\usepackage{textcomp}
\usepackage{calrsfs}
\usepackage{paralist}
\usepackage{enumitem}
\usepackage{changebar}
\usepackage{todonotes}
\usepackage[super]{nth}

\usepackage{fontawesome}

\DeclareMathAlphabet{\pazocal}{OMS}{zplmf}{m}{n}
\newcommand{\mcal}[1]{\pazocal{#1}}
\newcommand{\bcal}[1]{\mathcal{#1}}

\newcommand{\rulename}[1]{\ensuremath{\mbox{\textsc{#1}}}}

\newcommand{\qt}[1]{``{#1}"}

\newcommand{\rTo}[1]{\xrightarrow{#1}}

\newcommand{\modif}[1]{{\color{red}#1}}

\newcommand{\true}{{\sf true}}
\newcommand{\false}{{\sf false}}

\newlength{\arrow}
\newcounter{sqindex}

 \newcommand{\rom}[1]{ \textup{(\lowercase\expandafter{\romannumeral#1})}}
 
\usepackage[ruled]{algorithm2e} 

\SetAlFnt{\small}
\SetAlCapFnt{\small}
\SetAlCapNameFnt{\small}
\SetAlCapHSkip{0pt}
\IncMargin{-\parindent}

\newcommand \Until      {\mathbin{\mcal{U}\kern-.1em}}
\newcommand \Release     {\mathbin{\mcal{R}\kern-.1em}}
\newcommand \WaitFor    {\mathbin{\mcal{W}}}
\newcommand \Since      {\mathbin{\mcal{S}\kern-.08em}}

\newcommand \g    {\mathsf{{G}\kern.08em}}
\newcommand \f    {\mathsf{{F}\kern.08em}}
\newcommand \UntilHat   {\mathbin{\LTLhat{\mcal{U}}\kern-.1em}}

\newcommand \SinceHat   {\mathbin{\LTLhat{\mcal{S}}\kern-.08em}}

\newcommand \Impl       {\mathbin{\rightarrow}}
\newcommand \Iff        {\mathbin{\leftrightarrow}}

\renewcommand \phi      {\varphi}

\newcommand \ltl        {\textsc{ltl}\xspace}
\newcommand \ctl				{\textsc{ctl}\xspace}
\newcommand \ltal       {\textsc{ltol}\xspace}
\newcommand \pdl       {\textsc{pdl}\xspace}
\newcommand \pspace       {\textsc{pspace}\xspace}
\newcommand \expspace       {\textsc{expspace}\xspace}
\newcommand \nlogspace       {\textsc{nlogspace}\xspace}
\newcommand \nlog       {\textsc{nlog}\xspace}
\newcommand \B        	{\mcal{B}\xspace}
\newcommand \bb        	{{\Bbb B}\xspace}

\newcommand{\set}[1]{\{{#1}\}}
\newcommand{\conf}[1]{\langle{#1}\rangle}

\newcommand{\band}[3]{\bigwedge\limits_{#1}^{#2}{#3}}
\newcommand{\bor}[3]{\bigvee\limits_{#1}^{#2}{#3}}

\newcommand{\bpcup}[2]{\biguplus\limits_{#1}{#2}}
\newcommand{\func}[3]{{#1}:{#2}\rightarrow{#3}}
\newcommand{\pfunc}[3]{{#1}^{#2}_{#3}}
\newcommand{\such}{{\bf\mathsf{s.t.}}\ }
\newcommand{\toall}{\star}
\newcommand{\cstate}[2]{{s}^{#1}_{#2}}
\newcommand{\sstate}[1]{{s}_{#1}}
\newcommand{\vardom}[1]{\mathsf{Dom}({#1})}
\newcommand{\vdom}{\vardom{v}}
\newcommand{\compdom}[1]{\prod_{v\in V_{#1}}\vdom}

\newcommand{\sdom}{\prod_{i}\compdom{i}}

\newcommand{\tuple}[1]{\left(#1\right)}
\newcommand{\anglebr}[1]{\left [#1\right] }
\def\<#1>{\mathinner{\langle#1\rangle}}
\newcommand{\bnxt}[1]{\Biggl<{#1}\Biggr>}
\newcommand{\pred}{\pi}
\newcommand{\sysvar}{\mathcal{V}}

\newcommand{\chan}{ch}
\newcommand{\schan}{\rulename{ch}} 

\newcommand{\id}{i}

\newcommand{\cv}{cv}
\newcommand{\scv}{\rulename{cv}} 

\newcommand{\dat}{d}
\newcommand{\datfun}{{\bf d}}
\newcommand{\sdat}{\rulename{d}} 

\newcommand{\coma}{,\ }

\newcommand{\trans}{\mcal{T}}

\newcommand{\exis}[1]{\bullet^{\exists}{#1}}
\newcommand{\all}[1]{\bullet^{\forall}{#1}}
\newcommand{\nxt}[1]{\langle{#1}\rangle}
\newcommand{\alws}[1]{[{#1}]}

\newcommand{\tracevar}[1]{\rho_{#1}}
\newcommand{\transmain}{\delta_{\phi}}

\newcommand{\transaux}{f}
\newcommand{\size}[1]{|{#1}|}
\newcommand{\length}[1]{||{#1}||}
\newcommand{\Exp}[1]{2^{#1}}
\newcommand{\dexp}[1]{2^{2^{#1}}}
\newcommand{\bigo}[1]{\bcal{O}({#1})}
\newcommand{\lang}[1]{L_{\omega}({#1})}

\newcommand{\msf}[1]{\mathsf{#1}}

\newcommand{\rcp}{{{\rulename{ReCiPe}}}\xspace}

\newcommand{\comment}[1]{}

\newcommand{\keep}{\mbox{\sc keep}}
\newcommand{\agent}[1]{\rulename{{#1}}}
\newcommand{\lineagent}{\agent{Line}}

\newcommand{\typecvar}{{\scriptstyle@\msf{type}}}
\newcommand{\assigncvar}{{\scriptstyle@\msf{asgn}}}
\newcommand{\readycvar}{{\scriptstyle@\msf{rdy}}}
\newcommand{\lnkcvar}{{\scriptstyle@\msf{lnk}}}

\newcommand{\msgdvar}{\rulename{msg}}
\newcommand{\lnkdvar}{\rulename{lnk}}
\newcommand{\nodvar}{\rulename{no}}

\newcommand{\stlvar}{\msf{st}}
\newcommand{\lnklvar}{\msf{lnk}}
\newcommand{\prdlvar}{\msf{prd}}
\newcommand{\stagelvar}{\msf{stage}}
\newcommand{\stagebvar}{\msf{step}}
\newcommand{\typelvar}{\msf{ltype}}
\newcommand{\assignlvar}{\msf{lasgn}}
\newcommand{\readylvar}{\msf{lrdy}}
\newcommand{\typebvar}{\msf{btype}}
\newcommand{\assignbvar}{\msf{basgn}}
\newcommand{\readybvar}{\msf{brdy}}
\newcommand{\nolvar}{\msf{no}}
\newcommand{\val}[1]{\rulename{{\scriptsize\tt{#1}}}}
\newcommand{\taval}{\val{t1}}
\newcommand{\tbval}{\val{t2}}
\newcommand{\tcval}{\val{t3}}
\newcommand{\teamval}{\val{team}}
\newcommand{\assembleval}{\val{asmbl}}
\newcommand{\formval}{\val{form}}
\newcommand{\localval}{\val{local}}
\newcommand{\pendval}{\val{pnd}}
\newcommand{\enval}{\val{end}}

\newcommand{\startval}{\val{strt}}

\newcommand{\listen}{\rulename{ls}}
\newcommand\smallo{
  \mathchoice
    {{\scriptscriptstyle\mcal{O}}}
    {{\scriptscriptstyle\mcal{O}}}
    {{\scriptscriptstyle\mcal{O}}}
    {\scalebox{.3}{$\scriptscriptstyle\mcal{O}$}}
  }

\newcommand{\obs}[1]{{\msf{obs}({#1})}\xspace}
\usepackage{geometry}
\begin{document}

\title{Modelling and Verification of Reconfigurable Multi-Agent
  Systems}
\thanks{This research is funded by the ERC consolidator
    grant D-SynMA under the European Union's Horizon 2020 research and
   innovation programme (grant agreement No 772459); and by the Swedish research council grants:  SynTM (No. 2020-03401) and VR project (No. 2020-04963).The authors
   would like to thank Giuseppe Perelli who was a co-author on an
   earlier version of this article. }


\author{Yehia Abd Alrahman \and Nir Piterman 
}

\address{University of Gothenburg}


\maketitle

\begin{abstract}
	We propose a formalism to model and reason about reconfigurable multi-agent systems.
	In our formalism, agents interact and communicate in different modes
	so that they can pursue joint tasks; agents may dynamically synchronize, exchange data, adapt their behaviour, and reconfigure their communication interfaces.
	Inspired by existing multi-robot systems, we represent a 
	system as a set of agents (each with local state),
executing independently and
only influence each other by means of message exchange. 
Agents are able to sense their local states and
partially their surroundings.
	We extend \ltl to be able to reason explicitly about the intentions of agents in the interaction and their communication protocols. 
	We also study the complexity of satisfiability and model-checking of this extension.
\end{abstract}

\section{Introduction}\label{sec:intro}

In recent years formal modelling of multi-agent systems (MAS) and their
analysis through model checking has received much attention~\cite{Woo02,LQR17}.
Several mathematical formalisms have been suggested to represent the
behaviours of such systems and to reason about the strategies that
agents exhibit \cite{LQR17,AHK02}. 
For instance, 
modelling languages, such as RM \cite{AH99b,GHW17} and ISPL \cite{LQR17}, are used to
enable efficient analysis by representing these systems through the usage of
BDDs.
Temporal logics have been also extended and adapted (e.g., with Knowledge support
\cite{FHMV95} and epistemic operators \cite{GochetG06}) specifically to support
multi-agent modelling \cite{GL13}.
Similarly, logics that support reasoning about the intentions and
strategic abilities of such agents have been used and extended
\cite{CHP10,MMPV14}.

These works are heavily influenced by the formalisms used for verification
 (e.g., Reactive Modules~\cite{AH99b,AG04},
concurrent game structures~\cite{AHK02}, and interpreted
systems~\cite{LQR17}).
They 
rely on shared memory to implicitly model
interactions.
It is generally agreed that explicit message passing is more
appropriate to model interactions among distributed agents because of its
scalability~\cite{cougaar,mno}.
However, the mentioned 
formalisms trade the advantages of message passing for abstraction,
and abstract message exchange by controlling the visibility of state
variables of the different agents.
Based on an early result, where a compilation from shared memory to message
passing was provided \cite{stop}, it was believed that a shared memory model 
is a higher level abstraction of distributed systems. 
However, this result holds only in specific cases and
under assumptions that practically proved to be unrealistic, see~\cite{mm}.
Furthermore, the compositionality of shared memory approaches is 
limited and the supported interaction interfaces are in general not
very flexible~\cite{BBDM19}. 
Alternatively, message passing formalisms~\cite{pi1} are very
compositional and support flexible interaction interfaces.
However, unlike shared memory formalisms, they do not accurately
support awareness capabilities, where an agent may instantaneously
inspect its local state and adapt its behaviour while interacting.
The reason is that 
  they model agents as
  mathematical expressions over interaction operators. Thus the
  state of an agent is implicit in the structure of the expression.

To combine the benefits of both approaches recent developments~\cite{APUS,mm} 
suggest adopting hybrids, that accurately represent actual distributed systems, e.g., ~\cite{rdma,MathewsVCOBD15}. 
We propose a hybrid formalism to model and reason about
distributed multi-agent systems.
A system is represented as a set of agents (each with local state),
executing concurrently and
only interacting by message exchange. 
Inspired by multi-robot systems, 
e.g., Kilobot~\cite{kilobot} and Swarmanoid~\cite{swarm},
agents are additionally able to sense their local states and
partially their surroundings.
Interaction is driven by message passing following 
the interleaving semantics of~\cite{pi1};
in that only one agent may send a message at a time while other agents
may react to it.
To support meaningful interaction among agents \cite{WK16},
messages are not mere synchronisations, but carry data that might
be used to influence the behaviour of receivers. 

Our message exchange is adaptable and reconfigurable.
Thus, agents determine how to communicate and with whom.
Agents interact on links that change their utility
based on the needs of interaction at a given stage.
Unlike existing message-passing mechanisms, which use static 
notions of network connectivity to establish interactions, our mechanisms 
allow agents to specify receivers 
using logical formulas.
These formulas are interpreted over the evolving local
states of the different agents and thus provide a natural way to establish 
reconfigurable interaction interfaces (for example, limited range
communication \cite{MathewsVCOBD15}, messages destined for particular agents \cite{info19}, etc.).

The advantages of our formalism are threefold. We provide more realistic
 models that are close to their distributed implementations, and how actual 
 distributed MAS are developed, e.g.,~\cite{swarmdesign}. We provide a 
 modelling convenience for high level interaction features of MAS
  (e.g., coalition formation, collaboration, self-organisation, etc),
  that otherwise have to be hard-coded tediously in existing formalisms.  
Furthermore, we decouple the individual behaviour of agents from their interaction
protocols to facilitate reasoning about either one separately.
%

In addition, we extend \ltl to characterise messages and their
targets.
This way we allow
reasoning about the intentions of agents in communication.
Our logic can refer directly to the interaction protocols.
Thus the interpretation of a formula 
incorporates information
about the causes of assignments to variables and the flow of the interaction
protocol.
We also study the complexity of satisfiability and
Model-checking for our logic.

This article is an extended and revised version of the conference
paper presented in~\cite{rcp}. The major extensions in this article
consist of: \rom{1} a compositional and enumerative semantic
definition of the proposed formalism, that coincides with the early
symbolic one. The new definition facilitates reasoning about the
individual behaviour of agents and their compositions with others. For
this purpose, we defined a parallel composition operator with
reconfigurable broadcast and multicast semantics. Thus, the definition
is not only intuitive, but can also be used to reason about models
under  \emph{open-world assumption}; \rom{2} a major improvement on
our early results~\cite{rcp} regarding satisfiability and model
checking, that were computed in an \expspace upper bound. Here, we
provide a novel automata construction that permits \pspace analysis,
matching the lower bound. Thus, this part is majorly rewritten and
improved. Moreover, we enhance the presentation of the different parts
of the article and provide the proofs of all results.  

The structure of this article is as follows: In Sect.~\ref{sec:overview}, we informally 
present our formalism and motivate our design choices. In
Sect.~\ref{sec:bck}, we give the necessary background and in Sect.~\ref{sec:cts} we present the compositional semantic definition.
In Sect.~\ref{sec:model}   
 we introduce the formalism both in terms of enumerative and symbolic semantics, and we prove that they coincide. In Sect.~\ref{sec:exp}, we present a non-trivial case study to show the distinctive features of our formalism. 
In Sect.~\ref{sec:logic} we discuss our extension to LTL and provide efficient decision procedures to check both satisfiability and model checking in polynomial space. Finally, 
in Sect.~\ref{sec:conc} we discuss our concluding remarks.  

\section{An informal overview}\label{sec:overview}
We use a collaborative-robot scenario to informally illustrate the
 distinctive features of our formalism and we later formalise it in Section~\ref{sec:exp}.
The scenario is based on Reconfigurable Manufacturing Systems (RMS)~\cite{rms1,rms2},
where assembly product lines coordinate autonomously with different
types of robots to produce products.

In our formalism, each agent has a local state
consisting of a set of variables whose values may change due to
either contextual conditions or side-effects of
interaction. The external behaviour of an agent is only represented by the messages it
exposes to other agents while the local one is represented by changes to its
state variables. These variables are initialised by initial conditions
and updated by send- and receive- transition relations.
In our example, a product-line agent initiates different production
procedures based on the assignment to its product variable $\qt{\msf{prd}}$,
which is set by the operator, while it controls the progress of its status variable
$\qt{\msf{st}}$ based on interactions with other robots.
Furthermore, a product-line agent is characterised:
(1) externally only by the recruitment and assembly messages it  
sends to other robots and
(2) internally by a sequence of assignments to its local variables.

Before we explain the send- and receive- transition relations and show
the dynamic reconfiguration of communication interfaces we need to
introduce a few additional features. 
We assume that there is an agreed set of \emph{channels/links} $\schan$ that includes a unique
broadcast channel $\star$.
Broadcasts have non-blocking send and blocking receive while
multicasts have blocking send and receive.
In a broadcast, receivers (if exist) may anonymously receive a
message when they are interested in its values and when they satisfy
the send guard.
Otherwise, the agent does not participate in the interaction either
because they cannot (do not satisfy the guard) or because they are not
interested (make an idle transition).
In multicast, all agents connected to the multicast channel must
participate to enable the interaction.
For instance, recruitment messages are broadcast because a line agent
assumes that there exist enough robots to join the team while assembly messages
are multicast because they require that the whole connected team is ready to
assemble the product.

Agents dynamically decide (based on local state) whether they can
use (i.e., connect-to) multicast channels while the broadcast channel is always available. 
Thus,
initially, agents may not be connected to any channel, except for the
broadcast one $\msf{\star}$.
These channels may be learned using broadcast messages
and thus a structured communication interface can be built at run-time, starting
from a (possibly) flat one. 

Agents use messages to send selected data and specify
how and to whom.
Namely, the values in a message specify what is exposed to the others;
the channel specifies how to coordinate with others; and a send
guard specifies the target.
Accordingly, each message carries an assignment to a set of agreed
\emph{data variables} $\sdat$, i.e., the exposed data;  a
channel $\msf{\chan}$; and a send guard $\pfunc{g}{s}{}$.
In order to write meaningful send guards, we assume a set of
\emph{common variables} $\scv$ that each agent stores locally and
assigns its own information (e.g., the type of agent, its location,
its readiness, etc.). Send guards are expressed in terms of conditions
on these variables and are evaluated per agent based on their assigned
local values. 
Send guards are parametric to the local state of the sender
and specify what assignments to common variables a potential receiver
must have.
For example, an agent may send a dedicated link name to a selected set
of agents by assigning a data variable in the communicated message and
this way a coalition can be built incrementally at run-time.
In our RMS, the send guard of the recruitment message specifies
the types of the targeted robots while the data values expose the
number of required robots per type and a dedicated multicast link to
be used to coordinate the production. 

Targeted agents may use incoming messages to update their
states, reconfigure their interfaces, and/or adapt their behaviour.
In order to do so, however, agents are equipped with receive guards $\pfunc{g}{r}{}$; 
that might be parametrised to local variables and channels, and thus
dynamically determine if an agent is connected to a given channel. 
The interaction among different agents is then derived based on send- and receive-
transition relations.
These relations are used to decide when to send/receive a message and what are the
side-effects of interaction. 
Technically, every agent has a send and a receive transition relation.
Both relations are parameterised by the state variables of the agent,
the data variables transmitted on the message, and by the channel name.
A sent message is interpreted as a joint transition between the send
transition relation of the sender and the receive transition relations of all
the receivers.
For instance, a robot's receive guard specifies that other than
the broadcast link it is also connected to a multicast link that matches the current value of
its local variable $\qt{\msf{lnk}}$. The robot then uses its receive transition relation
to react to a recruitment message, for instance, by assigning to its $\qt{\msf{lnk}}$ the
link's data value from the message.

Furthermore, in order to send a message the following has to happen.
The send transition relation of the sender must hold on: a given state of the sender,
a channel name, and an assignment to data variables.
If the message is broadcast,
all agents whose assignments to common variables satisfy the
send guard jointly receive the message, the others discard it.
If the message is multicast, all connected agents must satisfy the
send guard to enable the transmission (as otherwise they block the
message).
In both cases, 
sender and receivers execute their send- and receive-transition
relations jointly.
The local side-effect of the message takes into account the origin
local state, the channel, and the data.
In our example, a (broadcast) recruitment 
message is received by all robots that are not assigned to other teams (assigned ones  
discard it) and as a side effect they connect to a multicast
channel that is specified in the message.
A (multicast) assembly message can only be sent when the whole 
recruited team is ready to receive (otherwise the message is blocked)
and as a side effect the team proceeds to the next production stage. 

Clearly, the dynamicity of our formalism stems from the fact that we base
interactions directly over the evolving states of the different agents 
rather than over static notions of network connectivity as of existing approaches. 

\section{Transition Systems and Finite Automata}\label{sec:bck}

We unify notations and give the necessary background.
We introduce doubly-labeled transition systems and discrete systems
and show how to translate the former to the latter. 
We further introduce nondeterministic and alternating B\"uchi word automata.

\subsection{Transition Systems and Discrete Systems}
A \emph{Doubly-Labeled Transition System} (TS) is ${\mathcal{T}}=\langle \Sigma,\Upsilon,S,S_0,R,L\rangle$, where $\Sigma$ is a \emph{state
  alphabet}, $\Upsilon$ is a \emph{transition alphabet}, $S$ is 
a set of states, $S_0\subseteq S$ is a set of initial states, $R\subseteq
S\times \Upsilon \times S$ is a transition relation, and $L:S\rightarrow
\Sigma$ is a labeling function.

A \emph{path} of a transition system $\mathcal{T}$ is
a maximal sequence of states and transition labels
$\sigma=s_0,a_0,s_1,a_1,\ldots$ such that $s_0\in S_0$ and 
for every $j\geq 0$ we have $(s_i,a_i,s_{i+1})\in R$. 
We assume that for every state $s\in S$ there are $a\in \Upsilon$ and
$s'\in S$ such that $(s,a,s')\in R$. 
Thus, a sequence $\sigma$ is maximal if it is infinite.
If $|\Upsilon|=1$ then $\mathcal{T}$ is a \emph{state-labeled transition
  system} and if $|\Sigma|=1$ then $\mathcal{T}$ is a
\emph{transition-labeled transition system}. 

We introduce \emph{Discrete Systems} (DS) that represent state-labeled systems symbolically.
A DS is $\mathcal{D} = \langle \mathcal{V}, \theta, \rho \rangle$, where the components of
$\mathcal{D}$ are as follows:
\begin{itemize}[label={$\bullet$}, topsep=0pt, itemsep=0pt, leftmargin=10pt]
\item
  $\mathcal{V} = \set{v_1,...,v_n}$: A finite set of typed variables.  
  Variables range over discrete domains, e.g., Boolean or Integer.
A \emph{state} $s$ is an interpretation of  $\mathcal{V}$,
  i.e., if $D_v$ is the domain of $v$, then $s$ is in
  $\prod_{v_i\in \mathcal{V}} D_{v_i}$.

  We assume some underlying first-order language over $\mathcal{V}$ that
  includes (i) \emph{expressions} constructed from the variables in
  $\mathcal{V}$, (ii) \emph{atomic formulas} that are either Boolean variables
  or the application of different predicates to expressions, and
  (iii) \emph{assertions} that are first-order formulas constructed from
  atomic formulas using Boolean connectives or quantification of
  variables. 
  Assertions, also sometimes called \emph{state formulas}, characterize
  states through restriction of possible variable values in them.
\item
  $\theta$ : This is an
  assertion over $\mathcal{V}$ characterising all the initial states of the
  DS.  
  A state is called \emph{initial} if it satisfies $\theta$.
\item
  $\rho$ : A \emph{transition relation}. 
  This is an assertion $\rho(\mathcal{V}\cup\mathcal{V}')$, where $\mathcal{V}'$ is a primed
  copy of variables in $\mathcal{V}$. The transition relation $\rho$
  relates a state $s\in\Sigma$ to its \emph{$\mathcal{D}$-successors}
  $s'\in\Sigma$, i.e., 
  $(s,s')\models \rho$, where $s$ is an interpretation to
  variables in $\mathcal{V}$ and $s'$ is for variables in $\mathcal{V}'$.
\end{itemize}

The DS ${\mathcal{D}}$ gives rise to a state transition system $\mathcal{T}_{\mathcal{D}}=\langle\Sigma,\{1\},T,T_0,R\rangle$, where $\Sigma$ and $T$ are the
set of states of $\mathcal{T}_\mathcal{D}$, $T_0$ is the set of
initial states, and $R$ is the set of triplets $(s,1,s')$ such that
$(s,s')\models \rho$. Clearly, the paths of ${\mathcal{T}}_\mathcal{D}$ are exactly the
 paths of $\mathcal{D}$, but the size of $\mathcal{T}_\mathcal{D}$ is 
exponentially larger than the description of $\mathcal{D}$.

A common way to translate a DLTS into a DS, which we adapt
and extend below, would be to include 
additional variables that encode the transition alphabet.
Given such a set of variables $\mathcal{V}_\Upsilon$, an assertion
$\rho(\mathcal{V} \cup \mathcal{V}_\Upsilon \cup \mathcal{V}')$ characterises
the triplets $(s,\upsilon,s')$ such that $(s,\upsilon,s') \models
\rho$, where $s$ supplies the interpretation to $\mathcal{V}$, $\upsilon$
to $\mathcal{V}_\Upsilon$ and $s'$ to $\mathcal{V}'$.

%

\subsection{Finite Automata on Infinite Words}

We use the automata-theoretic approach to linear temporal logic
\cite{VW94}.
Thus, we translate temporal logic formulas to automata. We give here
the necessary background. 



 
For an alphabet $\Sigma$, the set $\Sigma^{\omega}$ is the set of infinite
sequences of elements from $\Sigma$. 
Given an alphabet $\Sigma$ and a set $D$ of directions, a 
\emph{$\Sigma$-labeled $D$-tree} is a pair $\tuple{T,\tau}$, where 
$T \subseteq D^*$ is a tree over $D$ and $\tau:T \rightarrow \Sigma$ maps
each node of $T$ to a letter in $\Sigma$. 
A {\em path} $\pi$ of a tree $T$ is a set $\pi\subseteq T$ such that
$\epsilon\in\pi$ and for every $x\in\pi$ either $x$ is a leaf in $T$
or there exists a unique $\gamma\in D$ such that $x\cdot \gamma\in \pi$.
For $\pi=\gamma_1\cdot \gamma_2\cdots$, we write $\tau(\pi)$ for
$\tau(\epsilon)\cdot \tau(\gamma_1)\cdot \tau(\gamma_1\gamma_2) \cdots$. 

For a finite set $X$, let ${\B}^+(X)$ be the set of positive
Boolean formulas over $X$ (i.e., Boolean formulas built from elements
in $X$ using $\wedge$ and $\vee$), where we also allow the formulas
$\true$ and $\false$.
For a set $Y\subseteq X$ and a formula $\theta\in{\B}^+(X)$, we
say that $Y$ {\em satisfies} $\theta$ iff assigning $\true$ to
elements in $Y$ and assigning $\false$ to elements in $X\setminus Y$
makes $\theta$ true.

\begin{definition}[Nondeterministic B\"uchi Word Automata (NBW)]
A NBW is 
$N=\langle \Sigma$, $Q$, $Q_{in}$, $\delta$, $F\rangle$, where
$\Sigma$ is an input alphabet, $Q$ is a finite set of states,
$\delta:Q \times \Sigma \rightarrow \Exp{Q}$ is a
transition function,
$Q_{in} \subseteq Q$ is a set of initial state,
and $F \subseteq Q$ specifies a B\"uchi acceptance condition.
\end{definition}

A run of a NBW $N$ on $w=\sigma_0\sigma_1\cdots\in\Sigma^{\omega}$ is
a sequence $r=q_0q_1\cdots\in Q^{\omega}$ such that $q_0\in Q_{in}$,
and for all $i\geq 0$ we have
$q_{i+1}\in\delta({q_i,\sigma_{i+1}})$. For a run $r$, we denote the
set of automaton states visited 
infinitely often by $inf(r)$, i.e., $inf(r)=\{q
~|~ q \mbox{ appears infinitely often in }r\}$.
A run is \emph{accepting} if $inf(r) \cap F \neq
\emptyset$.
The NBW $N$ accepts $w$ if there exists an accepting run of $N$ on
$w$.
We say that $w$ is in the language of $N$ and denote by $\lang{N}$ the set
of words accepted by $N$. 

\begin{definition}[Alternating B\"uchi Word Automata (ABW)]
An alternating B\"uchi word automaton is of the form
$A=\langle \Sigma$, $Q$, $q_{in}$, $\delta$, $F\rangle$, where
$\Sigma$ is the input alphabet, $Q$ is a finite set of states,
$\delta:Q \times \Sigma \rightarrow {\B}^+(Q)$ is a
transition function,
$q_{in} \in Q$ is an initial state,
and $F \subseteq Q$ specifies a B\"uchi acceptance condition.
\end{definition}

A run of an ABW $A$ on $w=\sigma_0\sigma_1\cdots$ is a $Q$-labeled
$D$-tree, $\tuple{T,\tau}$, where
$\tau(\epsilon)=q_{in}$ and, for every $x\in T$, we have
$\{\tau(x\cdot \gamma_1),\ldots, \tau(x\cdot \gamma_k)\} \models
\delta(\tau(x),\sigma_{|x|})$ where $\{x\cdot \gamma_1,\ldots, x\cdot
\gamma_k\}$ is the set of children of $x$.
A run of $A$ is accepting if all its 
infinite paths satisfy the acceptance condition. 
For a path $\pi$, let
$inf(\pi)=\{q ~|~ q \mbox{ appears infinitely often in }\tau(\pi)\}$.
A path $\pi$ is accepting if
$inf(\pi) \cap F \neq \emptyset$.
Thus, every infinite path in the run tree must visit the acceptance set $F$
infinitely often.
The ABW $A$ accepts $w$ if there exists 
an accepting run on $w$.
As before, we denote by $\lang{A}$ the set of words accepted by $A$.

We state the following well known results about Linear Temporal Logic (LTL), NBW, and ABW
(omitting the definition of LTL).

\begin{theorem}[\cite{VW94,Var94}]
For every LTL formula $\varphi$ of length $n$ there exist an ABW
$A_\varphi$ with $O(n)$ states such that $L(A_\varphi)=L(\varphi)$.
\label{vardiwolper}
\end{theorem}

\begin{theorem}[\cite{mh84}]\label{prop:aton}
For every ABW $A$ with $n$ states there is an
NBW $N$ such that $\lang{N}=\lang{A}$.
The number of states of $N$ is in $\Exp{\bigo{n}}$. 
\end{theorem}

%
%
%

\section{Channelled Transition Systems}\label{sec:cts}
In this section, we propose \emph{Channelled Transition System (CTS)}
to facilitate compositional modelling of interactive systems. Namely,
we extend the format of transition labels of Doubly-Labelled
Transition Systems to also specify the role of
the transition (i.e., send- or receive- message) and the used
communication channels. We define a parallel composition operator that
considers both broadcast and multicast semantics and we study its
properties.
The techniques to prove these results are rather standard.
However, we are not familiar with a setup that conveniently allows the
existence of transitions to depend on subscription to channels as we
suggest below.
\subsection{Channelled Transition Systems (CTS)}

A \emph{Channelled Transition System} (CTS) is
$\mathcal{T}=\langle C, \Sigma,\Upsilon,S, \allowbreak S_0,R,L,\rulename{ls}\rangle$,
where $C$ is a set of channels, including the broadcast channel
($\toall$), $\Sigma$ is a \emph{state alphabet}, $\Upsilon$ is a 
\emph{transition alphabet}, $S$ is a set of states, $S_0\subseteq S$
is a set of initial states, $R\subseteq S\times \Upsilon \times S$ is
a transition relation, $L:S\rightarrow \Sigma$ is a labelling
function, and $\rulename{ls}:S \rightarrow 2^C$ is a channel-listening function such that
for every $s\in S$ we have $\toall \in \rulename{ls}(s)$.
We assume that $\Upsilon = \Upsilon^+ \times \{!,?\} \times C$, for
some set $\Upsilon^+$. That is, every transition labelled with some
$\upsilon\in\Upsilon^+$ is either a message send ($!$) 
or a message receive ($?$) on some channel $c\in C$. 

A \emph{path} of a CTS $\mathcal{T}$ is a maximal
sequence of states and transition labels
$\sigma=s_0,a_0,s_1,a_1,\ldots$ such that $s_0\in S_0$ and  
for every $i\geq 0$ we have $(s_i,a_i,s_{i+1})\in R$.
As before, 
we assume that for every state $s\in S$ there exist $a\in \Upsilon$ and
$s'\in S$ such that $(s,a,s')\in R$. 
Thus, a sequence $\sigma$ is maximal if it is infinite.

\begin{remark} Note that the transition labels $a_i$ of a CTS's path
  $\sigma=s_0,a_0,s_1,\allowbreak  
a_1,\ldots$  range over both send (!) and receive (?)
transitions. Depending on the underlying semantics of the CTS, send
transitions may happen independently regardless of the existence of
receivers, e.g., in case of broadcast semantics. However, receive
transitions may only happen  jointly with some send transition. By
allowing CTS's paths to also range over receive transitions, we can
model every system as a collection of (open) systems that interact
through message exchange. That is, a receive transition in a system is
a hole that is closed/filled when composed with a send
transition from another system. A complete system (i.e., with filled
holes) is called a \emph{closed system}.  

The analysis in this article considers closed systems where a system
path ranges over send transitions only. In other word, we only
consider the messages exchanged within the system under consideration.  

The parallel composition of systems is defined below.  
\end{remark}

\begin{definition}[Parallel Composition]\label{def:par}
Given two CTS
$\mathcal{T}_i=\langle C_i, \Sigma_i,\Upsilon_i,S_i, \allowbreak
S^i_0,R_i,L_i,\listen^i\rangle$, where $i\in \{1,2\}$ their composition
$\mathcal{T}_1\parallel \mathcal{T}_2$ is the following CTS
$\mathcal{T}=\langle C, \Sigma,\Upsilon,S, \allowbreak
S_0,R,L,\listen\rangle$,
where the components of $\mathcal{T}$ are:
\begin{itemize}
\item
  $C = C_1 \cup C_2$
\item
  $\Sigma = \Sigma_1 \times \Sigma_2$
\item
  $\Upsilon = \Upsilon^1 \cup \Upsilon^2$
\item
  $S = S_1\times S_2$
\item
  $S_0 = S_0^1 \times S_0^2$
\item
  $R = $

  { \small
    $$
    \begin{array}{l r}
  \left \{((s_1,s_2),(\upsilon,!,c),(s'_1,s'_2)) \left |~
  \begin{array}{l r }
    (s_1,(\upsilon,!,c),s'_1) \in R_1 \mbox{ and }
    (s_2,(\upsilon,?,c),s'_2) \in R_2 & \mbox{ or} \\[2ex] 
    (s_1,(\upsilon,?,c),s'_1) \in R_1 \mbox{ and }
    (s_2,(\upsilon,!,c),s'_2) \in R_2 & \mbox{ or} \\[2ex] 
    (s_1,(\upsilon,!,c),s'_1) \in R_1, c\notin \listen^2(s_2), \mbox{ and }
    s_2=s'_2 & \mbox{ or} \\[2ex] 
    c\notin \listen^1(s_1), s_1=s'_1, \mbox{ and } 
    (s_2,(\upsilon,!,c),s'_2) \in R_2 
  \end{array}
  \right . \right \} & \cup
  \\[6ex]
   \left \{((s_1,s_2),(\upsilon,?,c),(s'_1,s'_2)) \left |~
  \begin{array}{l r }
    (s_1,(\upsilon,?,c),s'_1) \in R_1 \mbox{ and }
    (s_2,(\upsilon,?,c),s'_2) \in R_2 & \mbox{ or} \\[2ex] 
    (s_1,(\upsilon,?,c),s'_1) \in R_1, c\notin \listen^2(s_2), \mbox{ and }
    s_2=s'_2 & \mbox{ or} \\[2ex] 
    c\notin \listen^1(s_1),  s_1=s'_1, \mbox{ and }
    (s_2,(\upsilon,?,c),s'_2) \in R_2 
  \end{array}
  \right . \right \}  & \cup\\[6ex]
  \left \{((s_1,s_2),(\upsilon,\gamma,\toall),(s'_1,s'_2)) \left |~ 
  \begin{array}{l r }
    \begin{array}{l r }
    \gamma\in\{ !,?\}, (s_1,(\upsilon,\gamma,\toall),s'_1)\in R_1, s_2=s_2' \mbox{ and }\\
    \forall s_2'' ~.~ (s_2,(\upsilon,?,\toall),s''_2) \notin R_2 &
    \hfill \mbox{ or} \\[2ex]
    \multicolumn{2}{l}{\gamma\in\{ !,?\}, s_1=s_1',
    \forall s_1'' ~.~ (s_1,(\upsilon,?,\toall),s''_1) \notin R_1\ \mbox{and}}\\
    (s_2,(\upsilon,\gamma,\toall),s'_2)\in R_2, 
    \end{array} 
  \end{array}
  \right . \right \}

  \end{array}
  $$
  }
  
\item
  $L(s_1,s_2) = (L_1(s_1),L_2(s_2))$
\item
  $\listen(s_1,s_2) = \listen^1(s_1)\cup \listen^2(s_2)$
\end{itemize}   
\end{definition}

The transition relation $R$ of the composition defines two modes of
interactions, namely multicast and broadcast.
In both interaction modes, the composition $\mathcal{T}$ sends a
message $(\upsilon,!,c)$ on channel $c$ (i.e.,
$((s_1,s_2),(\upsilon,!,c),(s'_1,s'_2))\in R$) if either
$\mathcal{T}_1$ or $\mathcal{T}_2$ is able to generate this message,
i.e, $(s_1,(\upsilon,!,c),s'_1)\in R_1$ or
$(s_2,(\upsilon,!,c),s'_2)\in R_2$.

Consider the case of a multicat channel.
A multicast is blocking. Thus, a multicast message is sent if either
it is received or the channel it is sent on is not listened to.
Suppose that a message originates from $\mathcal{T}_1$, i.e.,
$(s_1,(\upsilon,!,c),s'_1)\in R_1$.
Then, $\mathcal{T}_2$ must be able to
either receive the message or, in the case that $\mathcal{T}_2$ does
not listen to the channel, discard it.
CTS $\mathcal{T}_2$ receives if $(s_2,(\upsilon,?,c),s'_2)\in R_2$.
It discards if $c\notin\listen^2(s_2)$ and $s_2=s'_2$.
The case of $\mathcal{T}_2$ sending is dual.
Note that $\mathcal{T}_2$ might be a composition of other
CTS(s), say $\mathcal{T}_2=\mathcal{T}_3\|\mathcal{T}_4$. In this
case, $\mathcal{T}_2$ listens to channel $c$ if at least one of
$\mathcal{T}_3$ or $\mathcal{T}_4$ is listening.
That is, it could be that either $c\in(\listen(s_3)\cap\listen(s_4))$, 
$c\in(\listen(s_2)\backslash\listen(s_3))$, or
$c\in(\listen(s_2)\backslash\listen(s_4))$.
In the first case, both must receive the message.
In the latter cases, the listener receives and the non-listener
discards.
Accordingly,
when a message is sent by one system, it is propagated to all other
connected systems in a joint transition.
A multicast is indeed blocking because a connected system cannot
discard an incoming message on a channel it is listening to.
More precisely, a joint transition
$((s_1,s_2),(\upsilon,!,c),(s'_1,s'_2))$ where $c\in\listen(s_2)$
requires that $(s_2,(\upsilon,?,c),s'_2)$ is
supplied. In other words, message sending is blocked until all
connected receivers are ready to participate in the interaction.  

Consider now a broadcast.
A broadcast is non-blocking.
Thus, a broadcast message is either received or discarded.
Suppose that a message originates from $\mathcal{T}_1$, i.e.,
$(s_1,(\upsilon,!,\toall),s'_1)\in R_1$.
If $\mathcal{T}_2$
is receiving, i.e., $(s_2,(\upsilon,?,\toall),s'_2)\in R_2$ the
message is sent.
However, 
by definition, we have that $\toall\in\listen(s)$ for every $s$ in a
CTS.
Namely, a system may not disconnect the broadcast channel
$\toall$.
For this reason, the last part of the transition relation $R$
defines a special case for handling (non-blocking)
broadcast.
Accordingly, a joint transition
$((s_1,s_2),(\upsilon,\gamma,\toall),(s'_1,s'_2))\in R$ where
$\gamma\in\set{!,?}$ is always possible and may not be blocked by any
receiver.
In fact, if ($\gamma=\ !$) and $(s_1,(\upsilon,!,\toall),s'_1)\in R_1$
then the joint transition is possible whether
$(s_2,(\upsilon,?,\toall),s'_2)\in R_2$ or not.
In other words, a broadcast can happen even if there are no
receivers.
Furthermore, if ($\gamma=\ ?$) and
$(s_1,(\upsilon,?,\toall),s'_1)\in R_1$ then also the joint transition
is possible regardless of the other participants. In other words, a
broadcast is received only by interested participants.  

\subsection{Properties of Parallel Composition}

Our parallel composition is commutative and associative.
Furthermore, it supports non-blocking broadcast and blocking multicast
semantics as stated in the following: 
\begin{lemma}[Commutativity and Associativity]\label{lem:com+assoc}
Given two CTS
$\mathcal{T}_1$ and $\mathcal{T}_2$ we have that:
\begin{itemize}
\item $\|$ is commutative: $\mathcal{T}_1\|\mathcal{T}_2=\mathcal{T}_2\|\mathcal{T}_1$;
\item $\|$ is associative: $(\mathcal{T}_1\|\mathcal{T}_2) \| \mathcal{T}_3=\mathcal{T}_1\|(\mathcal{T}_2 \| \mathcal{T}_3)$.
\end{itemize}
\end{lemma}
\begin{proof} We prove each statement separately. In both statement, the proof proceeds by case analysis on the joint transition.
\begin{description}
\item[\bf ($\|$ is commutative):] we consider all possible joint transitions from  $\mathcal{T}_1\|\mathcal{T}_2$ and we show that they have corresponding transitions in $\mathcal{T}_2\|\mathcal{T}_1$ and vice versa. We only show one direction and the other direction follows in a similar way. 
\begin{itemize}
\item {\bf Case $\tuple{\upsilon,!,c}$:} By Def.~\ref{def:par}, we have that $((s_1,s_2),\tuple{\upsilon,!,c}, (s'_1,s'_2))\in R_{12}$ ($R_{12}$ is the transition relation of $\mathcal{T}_1\|\mathcal{T}_2$) in the following cases: 
\begin{enumerate}
\item  $(s_1,\tuple{\upsilon,!,c},s'_1)\in R_1$ and $(s_2,\tuple{\upsilon,?,c},s'_2)\in R_2$:  It follows that $((s_2,s_1),\tuple{\upsilon,!,c}, (s'_2,s'_1))\in R_{21}$ ($R_{21}$ is the relation of $\mathcal{T}_2\|\mathcal{T}_1$) is derivable because $(s_2,\tuple{\upsilon,?,c},s'_2)\in R_2$ and $(s_1,\tuple{\upsilon,!,c},s'_1)\in R_1$;

\item $(s_1,\tuple{\upsilon,!,c},s'_1)\in R_1$ and $c\notin\listen^{2}(s_2)$ and ($s_2=s'_2$): It follows that $((s_2,s_1),\tuple{\upsilon,!,c}, (s'_2,s'_1))\in R_{21}$ is derivable because $c\notin\listen^{2}(s_2)$ and ($s_2=s'_2$) and $(s_1,\tuple{\upsilon,!,c},s'_1)\in R_1$;

\item $(s_1,\tuple{\upsilon,!,\toall},s'_1)\in R_1$ and ($s_2=s'_2$) and $\forall s''_2. (s_2,\tuple{\upsilon,?,c},s''_2)\notin R_2$: It follows that\\ $((s_2,s_1),\tuple{\upsilon,!,\toall}, (s'_2,s'_1))\in R_{21}$ is derivable because ($s_2=s'_2$) and $\forall s''_2. (s_2,\tuple{\upsilon,?,c},s'_2)\notin R_2$ and $(s_1,\tuple{\upsilon,!,c},s'_1)\in R_1$;

\item  $(s_2,\tuple{\upsilon,!,c},s'_2)\in R_2$ and $(s_1,\tuple{\upsilon,?,c},s'_1)\in R_1$:  This is the symmetric case of case (1);

\item $(s_2,\tuple{\upsilon,!,c},s'_2)\in R_2$ and $c\notin\listen^{1}(s_1)$ and ($s_1=s'_1$): This is the symmetric case of case (2);

\item $(s_2,\tuple{\upsilon,!,\toall},s'_2)\in R_2$ and ($s_1=s'_1$) and $\forall s''_1. (s_1,\tuple{\upsilon,?,c},s''_1)\notin R_1$: This is the symmetric case of case (3).
\end{enumerate}

\item {\bf Case $\tuple{\upsilon,?,c}$:} By Def.~\ref{def:par}, we have that $((s_1,s_2),\tuple{\upsilon,?,c}, (s'_1,s'_2))\in R_{12}$ in the following cases:
\begin{enumerate}
\item  $(s_1,\tuple{\upsilon,?,c},s'_1)\in R_1$ and $(s_2,\tuple{\upsilon,?,c},s'_2)\in R_2$:  It follows that $((s_2,s_1),\tuple{\upsilon,?,c}, (s'_2,s'_1))\in R_{21}$ is derivable for the same reason;

\item $(s_1,\tuple{\upsilon,?,c},s'_1)\in R_1$ and $c\notin\listen^{2}(s_2)$ and ($s_2=s'_2$): It follows that $((s_2,s_1),\tuple{\upsilon,?,c}, (s'_2,s'_1))\in R_{21}$ is derivable because $c\notin\listen^{2}(s_2)$ and ($s_2=s'_2$) and $(s_1,\tuple{\upsilon,?,c},s'_1)\in R_1$;

\item $(s_2,\tuple{\upsilon,?,c},s'_2)\in R_2$ and $c\notin\listen^{1}(s_1)$ and ($s_1=s'_1$): This is the symmetric case of case (2);

\item $(s_1,\tuple{\upsilon,?,\toall},s'_1)\in R_1$ and ($s_2=s'_2$) and $\forall s''_2. (s_2,\tuple{\upsilon,?,c},s''_2)\notin R_2$: It follows that\\ $((s_2,s_1),\tuple{\upsilon,?,\toall}, (s'_2,s'_1))\in R_{21}$ is derivable because ($s_2=s'_2$) and $\forall s''_2. (s_2,\tuple{\upsilon,?,c},s''_2)\notin R_2$ and $(s_1,\tuple{\upsilon,?,c},s'_1)\in R_1$;

\item $(s_2,\tuple{\upsilon,?,\toall},s'_2)\in R_2$ and ($s_1=s'_1$) and $\forall s''_1. (s_1,\tuple{\upsilon,?,c},s''_1)\notin R_1$: This is the symmetric case of case (3).

\end{enumerate}

\end{itemize}

\item[\bf ($\|$ is associative):] we consider all possible joint transitions from  $(\mathcal{T}_1\|\mathcal{T}_2) \| \mathcal{T}_3$ and we show that they have corresponding transitions in $\mathcal{T}_1\|(\mathcal{T}_2 \| \mathcal{T}_3)$ and vice versa. We only show one direction and the other direction follows in a similar way. 
\begin{itemize}
\item {\bf Case $\tuple{\upsilon,!,c}$:} By Def.~\ref{def:par}, we have $(((s_1,s_2),s_3),\tuple{\upsilon,!,c}, ((s'_1,s'_2),s'_3))\in R_{(12)3}$ in the following cases: 
\begin{enumerate}
\item  $((s_1,s_2),\tuple{\upsilon,!,c},(s'_1,s'_2))\in R_{12}$ and $(s_3,\tuple{\upsilon,?,c},s'_3)\in R_3$:  As before, there are six cases for $((s_1,s_2),\tuple{\upsilon,!,c},(s'_1,s'_2))\in R_{12}$, we only consider the case  ($(s_1,\tuple{\upsilon,!,c},s'_1)\in R_1$ and $(s_2,\tuple{\upsilon,?,c},s'_2)\in R_2$); and other cases follow similarly. It follows that \\
$((s_1,(s_2,s_3)),$ $\tuple{\upsilon,!,c}, (s'_1,(s'_2,s'_3))\in R_{1(23)}$ where $(s_1,\tuple{\upsilon,!,c},s'_1)\in R_1$ and\\ $((s_2,s_3),\tuple{\upsilon,?,c}, (s'_2,s'_3))\in R_{23}$ such that $(s_2,\tuple{\upsilon,?,c}, s'_2)\in R_2$ and $(s_3,\tuple{\upsilon,?,c}, s'_3)\in R_3$ as required;
\item $((s_1,s_2),\tuple{\upsilon,!,c},(s'_1,s'_2))\in R_{12}$ and $c\notin\listen^{3}(s_3)$ and ($s_3=s'_3$): It follows that \\
$((s_1,(s_2,s_3)),\tuple{\upsilon,!,c}, (s'_1,(s'_2,s'_3))\in R_{1(23)}$ where $(s_1,\tuple{\upsilon,!,c},s'_1)\in R_1$ and\\ $((s_2,s_3),\tuple{\upsilon,?,c}, (s'_2,s'_3))\in R_{23}$ such that $(s_2,\tuple{\upsilon,?,c}, s'_2)\in R_2$ and  $c\notin\listen^{3}(s_3)$ and ($s_3=s'_3$);

\item $((s_1,s_2),\tuple{\upsilon,!,\toall},(s'_1,s'_2))\in R_{12}$ and ($s_3=s'_3$) and $\forall s''_3. (s_3,\tuple{\upsilon,?,\toall},$ $s'_3)\notin R_3$: It follows that 
$((s_1,(s_2,s_3)),\tuple{\upsilon,!,\toall}, (s'_1,(s'_2,s'_3))\in R_{1(23)}$ where $(s_1,\tuple{\upsilon,!,\toall},s'_1)\in R_1$ and\\ $((s_2,s_3),\tuple{\upsilon,?,\toall}, (s'_2,s'_3))\in R_{23}$ such that $(s_2,\tuple{\upsilon,?,\toall}, s'_2)\in R_2$ and  ($s_3=s'_3$) and $\forall s''_3. (s_3,\tuple{\upsilon,?,\toall},$ $s'_3)\notin R_3$;

\item  $((s_1,s_2),\tuple{\upsilon,?,c},(s'_1,s'_2))\in R_{12}$ and $(s_3,\tuple{\upsilon,!,c},s'_3)\in R_3$:  This is the symmetric case of case (1);

\item $(s_3,\tuple{\upsilon,!,c},s'_3)\in R_3$ and $c\notin\listen^{12}(s_1,s_2)$ and ($(s_1,s_2)=(s'_1,s'_2)$): This is the symmetric case of case (2);

\item $(s_3,\tuple{\upsilon,!,\toall},s'_3)\in R_3$ and $(s_1,s_2)=(s'_1,s'_2)$ and $\forall (s''_1,s''_2).$\\ $((s_1,s_2),\tuple{\upsilon,?,c},(s''_1,s''_2))\notin R_1$: This is  symmetric to case (3).
\end{enumerate}

\item {\bf Case $\tuple{\upsilon,?,c}$:} it follows similarly by case analysis on Def.~\ref{def:par}.
\end{itemize}

\end{description}

\end{proof}

\begin{lemma}[Non-blocking Broadcast]\label{lem:brdnblk}
Given a CTS $\mathcal{T}_1$ and for every other CTS $\mathcal{T}$, we
have that for every reachable state $(s_1,s)$ of
$\mathcal{T}_1\|\mathcal{T}$ the
following holds.  
\[(s_1,(\upsilon,!,\toall),s'_1)\in R_1\ \mbox{implies}\ ((s_1,s),(\upsilon,!,\toall),(s'_1,s'))\in R_{\mathcal{T}_1\|\mathcal{T}}\] 
\end{lemma}
\begin{proof} By Def.~\ref{def:par}, we have only two cases to derive $((s_1,s),(\upsilon,!,\toall),(s'_1,s'))\in R_{\mathcal{T}_1\|\mathcal{T}}$ given that $(s_1,(\upsilon,!,\toall),s'_1)\in R_1$. Note that, by definition, the condition $\toall\in\listen^{k}({s})$ always holds for any agent $k$ and in any state $s$. We show that when the channel is a broadcast $\toall$, the receiver does not play any role in enabling the transmission on the channel. In other words, it is only sufficient to have a sender to enable a broadcast at system level. More precisely, if $(s_1,(\upsilon,!,\toall),s'_1)\in R_1$ then we have the following:
\begin{itemize}
\item $((s_1,s),(\upsilon,!,\toall),(s'_1,s'))\in R_{\mathcal{T}_1\|\mathcal{T}}$ because $(s,(\upsilon,?,\toall),s')\in R$; or

\item $((s_1,s),(\upsilon,!,\toall),(s'_1,s'))\in R_{\mathcal{T}_1\|\mathcal{T}}$ because 
($s=s'$) and $\forall s''. (s,\tuple{\upsilon,?,\toall},s')\notin R_3$.
\end{itemize}
Namely, whether there exists a receiver or not, a broadcast can always happen (cannot be blocked).
\end{proof}
\begin{lemma}[Blocking Multicast]\label{lem:multi}
Given a CTS $\mathcal{T}_1$ and a multicast channel $c\in
C\backslash\set{\toall}$ such that $(s_1,(\upsilon,!,c),s'_1)\in R_1$,
then for every other CTS $\mathcal{T}$ we have that in every reachable
state $(s_1,s)$ of $\mathcal{T}_1\|\mathcal{T}$ the following holds.  
\[\begin{array}{lr}
((s_1,s),(\upsilon,!,c),(s'_1,s'))\in R_{\mathcal{T}_1\|\mathcal{T}} \ \mbox{iff}&\\[2ex]
 \qquad\qquad\qquad\qquad\qquad \qquad\qquad  \tuple{
  \begin{array}{rr}
  \multicolumn{2}{l}{c\in\listen(s)\ \mbox{and}\ (s,(\upsilon,?,c),s')\in R} \\[4pt]\mbox{or} &
 {  c\notin\listen(s)}
  \end{array}}
\end{array}
\] 
\end{lemma}
\begin{proof}
We show that it is not sufficient to only have a sender on a multicast channel $c$ to enable a send transition at system level. We require that all listening/connected agents to that channel being able to jointly receive the transmitted message. This also implies that if no one is listening then the transition can happen. By Def.~\ref{def:par}, there exist only two cases where $((s_1,s),(\upsilon,!,c),(s'_1,s'))\in R_{\mathcal{T}_1\|\mathcal{T}}$ given that $(s_1,(\upsilon,!,c),s'_1)\in R_1$.
More precisely, if $(s_1,(\upsilon,!,c),s'_1)\in R_1$ then we have the following:
\begin{itemize}
\item $((s_1,s),(\upsilon,!,c),(s'_1,s'))\in R_{\mathcal{T}_1\|\mathcal{T}}$ because $(s,(\upsilon,?,c),s')\in R$. This implies that $c\in\listen(s)$; or

\item $((s_1,s),(\upsilon,!,c),(s'_1,s'))\in R_{\mathcal{T}_1\|\mathcal{T}}$ because 
($s=s'$) and $c\notin\listen(s)$.
\end{itemize}
Namely, either all receivers jointly participate or no one is listening, as otherwise the multicast on $c$ is blocked. The other direction of the proof follows similarly.
\end{proof}

\section{\rcp: Reconfigurable Communicating Programs}\label{sec:model}

We formally present the \rcp communication
formalism and its main ingredients. We start by specifying agents (or programs)
and their local behaviours.
We give semantics to individual agents in terms of channelled
transition systems (CTS).
Therefore, we use the parallel composition operator in
Def.~\ref{def:par} to compose the individual behaviour of the
different agents to generate a global (or a system) one.

While the CTS semantics makes it clear what are the capabilities of
 individual agents and their interaction, it may not be the most
  convenient in order to mechanically
analyse large systems comprised of multiple agents.
Thus, we provide a symbolic semantics at system level using discrete
systems.
This second semantics enables efficient analysis by representing
\emph{closed} systems through the usage of BDDs.
We show that the two semantics (when restricted to \emph{closed}
systems) coincide.
The efficient analysis of \emph{open} \rcp systems is left as future
work. 

We assume that agents agree on a set of common variables $\scv$, a set
of data variables $\sdat$, and a set of  channels $\schan$ containing
the broadcast channel $\toall$.
As explained, common variables are variables that are owned
(separately) by all agents. The values of these variables may be
different in different agents.
The common variables are used in order to have a common langauge to
express properties that are interpretable on all agents (as either
true or false).

\begin{definition}[Agent]\label{def:comp}
An agent is $
A_{\id}=\langle V_{\id}\coma
f_{\id}\coma \pfunc{g}{s}{\id}\coma
\pfunc{g}{r}{\id},$ $\pfunc{\trans}{s}{\id}\coma \pfunc{\trans}{r}{\id},\theta_{\id}\rangle
$, where:
\begin{itemize}[label={$\bullet$}, topsep=0pt, itemsep=0pt, leftmargin=10pt]
\item
  $V_{\id}$ is a finite set of typed local variables, each
  ranging over a finite domain. A state $\cstate{\id}{}$ is an
  interpretation of $V_{\id}$, i.e., if $\mathsf{Dom}(v)$ is the
  domain of $v$, then $\cstate{\id}{}$ is an element in
  $\compdom{\id}$. We use $V'$ to denote the primed copy of $V$ and  
  $\mathsf{Id}_{\id}$ to denote the assertion $\bigwedge_{v\in
    V_{\id}}v=v'$. 

\item
  $\func{f_{\id}}{\scv}{V_{\id}}$ is a renaming function, associating
  common variables to local variables.
  We freely use the notation $f_{\id}$ for the assertion
  $\bigwedge_{\cv\in \scv}\cv=f_{\id}(\cv)$. 

\item
  $\pfunc{g}{s}{\id}(V_{\id}, \schan, \sdat, \scv)$ is a send guard
  specifying a condition on receivers. That is, the predicate, obtained
  from $\pfunc{g}{s}{\id}$ after assigning $\cstate{\id}{}$, $\chan$, and $\datfun$ (an
  assignment to $\sdat$)
  , which is checked against every receiver $j$ after
  applying $f_{j}$.

\item
  $\pfunc{g}{r}{\id}(V_{\id}, \schan)$ is a receive guard describing
  the connection of an agent to channel $\chan$. We let
  $\pfunc{g}{r}{\id}(V_{\id}, \toall)$ $ = \true$, i.e., every agent
  is always connected to the broadcast channel. We note, however, that
  receiving a broadcast message could have no effect on an agent. 
\item
  $\pfunc{\trans}{s}{\id}(V_{\id}, V'_{\id}, \sdat, \schan)$ is an assertion
  describing the send transition relation.
\item
  $\pfunc{\trans}{r}{\id}(V_{\id}, V'_{\id}, \sdat, \schan)$ is an assertion
  describing the receive transition relation.
  We assume that agents are broadcast input-enabled, i.e.,
  $\forall v, \datfun\ \exists v'\ \such$ $\pfunc{\trans}{r}{\id}(v, v',
  \datfun, \toall)$.

  In examples, we use $\keep(X)$ to denote that the variables $X$ are
  not changed by a transition (either send or receive).
  More precisely, $\keep(X)$ is equivalent to the following assertion
  $\bigwedge_{x\in X}x=x'$. 
  
\item
  $\theta_{\id}$ is an assertion on $V_{\id}$ describing the
  initial states, i.e., a state is initial if it satisfies
  $\theta_{\id}$.  
\end{itemize}
\end{definition}

Agents exchange messages. A message (that we shall call \emph{an observation}) is defined
by the channel it is sent on (\chan), the data it carries (\datfun), the sender identity (\id),
and the assertion describing the possible local assignments to
common variables of receivers ($\pi$). 
Formally:

\begin{definition}[Observation]\label{def:obsrv}
An observation is a tuple
$m=\tuple{\chan,\datfun,\id,\pred}$, where $\chan$ is a channel, $\datfun$ is
an assignment to $\sdat$, $\id$ is an identity,
and $\pred$ is a predicate over \scv.
\end{definition}
In Def.~\ref{def:obsrv} we interpret $\pred$ as a set of possible assignments to common
variables $\scv$.
In practice, $\pred$ is obtained from
$\pfunc{g}{s}{\id}(\cstate{\id}{},\chan,\datfun,\scv)$ for an agent
$\id$, where $\cstate{\id}{}\in\compdom{\id}$ and $\chan$ and $\datfun$
are the channel and assignment in the observation. 
We freely use $\pi$ to denote either a predicate over $\scv$ or its
interpretation, i.e., the set of variable assignments $c$ such that $c
\models \pi$. We also use $\pred(f^{-1}_{\id}(s_{\id}))$ to denote the assignment of $v\in\scv$ by
     $s_{\id}(f_{\id}(v))$ in $\pred$.

The semantics of an agent $A_i$ is the CTS $\mathcal{T}(A_i)$ defined
as follows.
\begin{definition}[Agent Semantics]\label{def:agentsem}
  Given an agent $A_i$ we define $\mathcal{T}(A_i)=
\langle C,\Sigma,\Upsilon,S,S_0,R,L,\rulename{ls}\rangle$, where the components
of $\mathcal{T}(A_i)$ are as follows.
\begin{itemize}
\item
  $C = \schan$
\item
  $\Sigma = \prod_{v\in V_i} \mathsf{Dom}(v)$, i.e., the set of states
  of $A_i$
\item
  $\Upsilon = \Upsilon^+\times \{!,?\} \times \schan$ and $\Upsilon^+ =
  2^{\sdat} \times K \times 2^{2^{\scv}}$
\item
  $S=\Sigma$
\item
  $S_0 = \{ s\in S ~|~ \theta_i(s) \}$
\item
  $R =$
  $$\begin{array}{l r}
  \{(s,(\datfun,i,\pi,!,c),s') ~|~ \pfunc{\trans}{s}{\id}(s,s',\datfun,c) \mbox{ and }
  \pi=\pfunc{g}{s}{\id}(s,c,\datfun)\} & \cup \\
  \{(s,(\datfun,\id',\pi,?,c),s') ~|~ \pfunc{\trans}{r}{\id}(s,s',\datfun,c), \id'\neq \id, c\in \rulename{ls}(s), \mbox{ and }
  \pred(f^{-1}_{\id}(s_{\id}))
  \} 
\end{array}
  $$
\item
  $L(s)=s$
\item
  $\rulename{ls}(s) = \{ c\in C ~|~ \pfunc{g}{r}{\id}(s,c)\}$
\end{itemize}
\end{definition}

Generally, the semantics of an agent is defined as an \emph{open} CTS
$\mathcal{T}(A_i)$. The transition alphabet $\Upsilon$ of
$\mathcal{T}(A_i)$ is the set of observations (as in
Def.~\ref{def:obsrv}) that are additionally labelled with either send
(!) or receive (?) symbols, corresponding to send and receive
transitions. Furthermore, in every state $s$, an agent is listening
to the set of channels in $\listen(s)$. Namely, all channels that
satisfy the agent's receive guard $\pfunc{g}{r}{\id}$ in state $s$. 
We give further intuition for the definition of the transition
relation $R$.

A triplet $(s,\upsilon,s') \in R$, where
$\upsilon=(\datfun,\id,\pred,\gamma,\chan)$, if the following holds:
\begin{itemize}[label={$\bullet$}]
\item
  Case ($\gamma=!$): Agent $\id$ is a sender and we have that
  $\pred=\pfunc{g}{s}{\id}(s_{\id},\chan, \datfun)$, i.e.,
  $\pred$ is obtained from $\pfunc{g}{s}{\id}$ by assigning the
  state of $\id$, the data variables assignment $\datfun$ and the channel $\chan$, and
  $\pfunc{\trans}{s}{\id}(s_{\id},s'_{\id},\datfun,\chan)$ evaluates to
  $\true$. 
\item
  Case ($\gamma=?$): Agent $\id$ is a receiver (potentially) accepting
  a message from another agent $\id'$ on channel $c$ and data
  $\datfun$ with a send guard $\pi$ such that $c\in\listen(s)$, 
  $\pred(f^{-1}_{\id}(s_{\id}))$, and
  $\pfunc{\trans}{r}{\id}(s_{\id},s'_{\id},\datfun,\chan)$.
  Note that the condition $\id'\neq\id$ is required to ensure that the
  message is sent by another agent.
  
\end{itemize}

Intuitively, if the agent $\id$ is the sender, it determines the
predicate $\pi$ (by assigning $s_{\id}$, $\datfun$, and $\chan$ in $\pfunc{g}{s}{\id}$)
and $\id$'s send transition $\pfunc{\trans}{s}{\id}$ is satisfied by assigning $s_{\id}$,
$s'_{\id}$,  $\datfun$, and $\chan$ to it.
That is, upon sending the message with $\datfun$ on channel $\chan$ the sender changes the state from
$s_\id$ to $s'_{\id}$.
If the agent $\id$ is the receiver, it must satisfy the
condition on receivers $\pi$ (when translated to its local copies of the
common variables), it must be connected to $\chan$ (according to
$\pfunc{g}{r}{\id}$), and it must have a valid receive transition
$\pfunc{\trans}{r}{\id}$ when reading the data sent in $\datfun$ on
channel $\chan$. 

Note that the semantics of an individual agent is 
totally decoupled from the semantics of how agents
interact. 
Thus, different interaction modes (or parallel composition operators)
can be adopted without affecting the semantics of individual agents.
In our case, we have chosen to implement broadcast as a non-blocking send and non-blocking receive and a multicast as a blocking send and receive.
However, if one chooses to do so, other composition operators could be defined. 
For example, a point-to-point composition would allow only two agents to communicate over a channel and would not allow send without receive.

A set of agents agreeing on the common variables $\scv$,
data variables $\sdat$, and channels $\schan$ define a \emph{system}.
We define a CTS capturing the
interaction and 
then give a DS-like symbolic representation of the same system.

Let $S_{\id}$=$\Pi_{v\in V_{\id}}\mathsf{Dom}(v)$ be the set of states of agent $\id$ and
$S=\Pi_{\id}S_{\id}$ be the set of states of the whole system.
Given an assignment $s\in S$ we denote by $s_{\id}$ the projection of 
$s$ on $S_{\id}$.

\begin{definition}[\rcp System as a CTS]\label{def:ts}
  Given a set $\{A_{\id}\}_{\id}$ of agents,
  a system is defined as the parallel composition of the CTS representations of all $A_{\id}$, i.e., a system is  a CTS of the form $\mathcal{T}=\|_{\id\in I}\mathcal{T}(A_{\id})$.
  
  A triplet $(s,\upsilon,s')$, where
  $\upsilon=(\datfun,i,\pi,!,c)$ is in the transition relation of the
  composed system $\mathcal{T}$ (according to Def.~\ref{def:par}), if
  the following conditions hold: 
  \begin{itemize}[label={$\bullet$}, topsep=0pt, itemsep=0pt, leftmargin=10pt]
  \item
    There exists a sender $\id$ such that
    $(s_i,(\datfun,i,\pi,!,c),s'_i)\in R_{\id}$.
    By Def.~\ref{def:agentsem}, we know that\\ $(s_i,(\datfun,i,\pi,!,c),s'_i)\in R_{\id}$ iff 
    $\pred=\pfunc{g}{s}{\id}(s_{\id},\chan, \datfun)$, i.e.,
    $\pred$ is obtained from $\pfunc{g}{s}{\id}$ by assigning the
    state of $\id$, the data variables assignment $\datfun$ and the channel $\chan$, and
    $\pfunc{\trans}{s}{\id}(s_{\id},s'_{\id},\datfun,\chan)$ evaluates to
    $\true$. 
  \item
    For every other agent $\id'$ we have that either:
    \begin{enumerate}
    \item
      $c\in\listen^{\id'}(s_{\id'})$ and
      $(s_{\id'},(\datfun,i,\pi,?,c),s'_{\id'})\in R_{\id'}$.
      By Def.~\ref{def:agentsem}, we know that $c\in\listen^{\id'}(s_{\id'})$ and  
      $(s_{\id'},(\datfun,\id,\pi,?,c),s'_{\id'})\in R_{\id'}$ iff
      $\pfunc{g}{r}{\id'}(s_{\id'},c)$,
      $\pred(f^{-1}_{\id'}(s_{\id'}))$, and
      $\pfunc{\trans}{r}{\id'}(s_{\id'},s'_{\id'},\datfun,\chan)$,
      all evaluate to $\true$;  
     
   \item
     $c\notin\listen^{\id'}(s_{\id'})$ and $s_{\id'}=s'_{\id'}$.
     By Def.~\ref{def:agentsem} this is equivalent to $\neg\pfunc{g}{r}{\id'}(s_{\id'},\chan)$;
     or 
     
   \item
     $\chan=\star$ and $s_{\id'}=s'_{\id'}$. By
     Def.~\ref{def:agentsem} this is equivalent to
     $\neg\pred(f^{-1}_{\id'}(s_{\id'}))$. 
    \end{enumerate}
  \end{itemize}
\end{definition}
Intuitively, a message $(\datfun,\id,\pi,!,c)$ labels a
transition from $s$ to $s'$ if the sender $\id$ determines the
predicate (by assigning $s_{\id}$, $\datfun$, and $\chan$ in $\pfunc{g}{s}{\id}$)
and the send transition of $\id$ is satisfied by assigning $s_{\id}$,
$s'_{\id}$, $\datfun$, and $\chan$ to it, i.e., the sender changes the state from
$s_\id$ to $s'_{\id}$ and sets the data variables in the observation
to $\datfun$. All the other agents either (a) satisfy this
condition on receivers (when translated to their local copies of the
common variables), are connected to $\chan$ (according to
$\pfunc{g}{r}{\id'}$), and perform a valid transition when reading the
data sent in $\datfun$ on $\chan$, (b) are not connected to $\chan$
(according to $\pfunc{g}{r}{\id'}$) and all their variables do not
change, or (c) the channel is a broadcast channel, the agent
does not satisfy the condition on receivers, and all their variables
do not change.

In order to facilitate symbolic analysis, 
we now define a symbolic version of $\|_{k\in K}\mathcal{T}(A_{k})$,
under closed world assumption. That is, we only focus on messages that
originate from the system under consideration. In fact, from an
external observer point of view, only message sending is observable
while reception cannot be observed. This notion of observability is
the norm in existing theories on group communication~\cite{pra,ene}. 
Thus, we consider the paths of $\|_{k\in K}\mathcal{T}(A_{k})$ that
are of the form $\sigma=s_0,a_0,s_1,a_1,\ldots$ such that $a_j$ is of
the form $(\datfun,i,\pi,!,c)$, $s_0\in S_0$ and   
for every $j\geq 0$ we have $(s_j,a_j,s_{j+1})\in R$.
Note that $(\datfun,i,\pi,!,c)$ coincides with our definition of an
observation $m$. 

Thus, let $\Upsilon$ be the set of possible observations in $\|_{k\in K}\mathcal{T}(A_{k})$.
That is, let $\schan$ be the set of channels, $\mathcal{D}$ the
product of the domains of variables in $\sdat$, $K$ the set of
agent identities, and $\Pi(\scv)$ the set of predicates over $\scv$
then
$\Upsilon \subseteq \schan \times \mathcal{D}\times K\times
\Pi(\scv)$.
In practice, we restrict attention to predicates in $\Pi(\scv)$ that
are obtained from $\pfunc{g}{s}{\id}(V_{\id},\schan,\sdat, \scv)$ by assigning
$V_{\id}$ (a state of the agent with identity $\id$), $\schan$, and $\sdat$.

Furthermore, we extend the format of the
allowed transitions in the classical definition of a discrete system from assertions over an extended set of variables
to assertions that allow quantification.
\begin{definition}[Discrete System]\label{def:sys}
Given a set $\{A_{\id}\}_{\id}$ of agents, a system is defined as follows: $
S=\conf{\sysvar\coma\rho\coma\theta}
$, where $\sysvar=\bpcup{\id}{}$ and $\theta=\band{\id}{}{\theta_{\id}}$ and a state of the system is in $\sdom$.  
The transition relation of the system is characterised as follows:
\[
\begin{array}{l}
\rho:\ \exists \chan\ \exists
\sdat\ \bor{k}{}{\pfunc{\trans}{s}{k}(V_k, V'_k, \sdat, \chan)} \wedge\\
\qquad
\band{j\neq k}{}{}
\tuple{
\exists \scv. f_j\wedge 
\tuple{\begin{array}{l r}
\multicolumn{2}{l}{\pfunc{g}{r}{j}(V_j, \chan)\wedge \pfunc{\trans}{r}{j}(V_j, V'_j, \sdat, \chan) \wedge\ \pfunc{g}{s}{k}(V_k, \chan, \sdat, \scv)}
\\[4pt]
\vee & 
\neg\pfunc{g}{r}{j}(V_j, \chan)\wedge \mathsf{Id}_j
\\[4pt]
\vee & 
\chan=\toall\wedge \neg \pfunc{g}{s}{k}(V_k, \chan, \sdat, \scv)\wedge \mathsf{Id}_j

\end{array}}}
\end{array}
\]
\end{definition}

The transition relation $\rho$ relates a system state
$\sstate{}$ to its successors $s'$ given an
observation $m=\tuple{\chan,\datfun,k,\pred}$. 
Namely, there exists an agent $k$ that sends a message with data $\datfun$
(an assignment to $\sdat$) with assertion $\pred$ (an assignment to
$\pfunc{g}{s}{k}$) on channel $\chan$ and all other agents are either (a)
connected, satisfy the send predicate, and participate in the
interaction, (b) not connected and idle, or (c) do not satisfy the
send predicate of a broadcast and idle.
That is, the agents satisfying $\pred$ (translated to their local state
by the conjunct $\exists \scv.f_j$) and connected to channel $\chan$ 
(i.e., $\pfunc{g}{r}{j}(\cstate{j}{}, \chan)$) get the
message and perform a receive transition. As a result of interaction,
the state variables of the sender and these receivers might be
updated.
The agents that are \emph{not connected} to the channel (i.e.,
$\neg\pfunc{g}{r}{j}(\cstate{j}{}, \chan)$) do not
participate in the interaction and stay still.
In case of broadcast, namely when sending
on $\toall$, agents are always connected and the set of receivers 
not satisfying $\pred$ (translated again as above) stay still.
Thus, a blocking multicast arises when a sender is blocked until all
\emph{connected} agents satisfy $\pred\wedge f_j$.
The relation ensures that, when sending on
a channel that is different from the broadcast channel $\toall$, the set
of receivers is the full set of \emph{connected} agents.
On the broadcast channel agents who do not satisfy the send
predicate do not block the sender.

The translation above to a transition system leads to a natural
definition of a trace, where the information about channels, data,
senders, and predicates is lost. We extend this definition to
include this information as follows:

\begin{definition}[System trace]\label{def:systrace}
A system trace is an infinite sequence 
$\rho=\sstate{0}m_0,\sstate{1}m_1,\dots$ of system states and
observations such that $\forall t\geq 0$:
$m_t=\tuple{\chan_t,\datfun_t,k,\pred_t}$,
$\pred_t=\pfunc{g}{s}{k}(\cstate{k}{t},\datfun_t, \chan_t)$, and: 
\[
\begin{array}{l c l }
    \tuple{\sstate{t}
  \coma\sstate{t+1}}  \models  
{\pfunc{\trans}{s}{k}(\cstate{k}{t}, \cstate{k}{t+1},
\datfun_t, \chan_t)}  \wedge \\
\qquad\qquad\qquad
\band{j\neq k}{}{}
\tuple{\exists \scv. f_j\wedge 
\tuple{\begin{array}{l r}
\multicolumn{2}{l}{\pfunc{g}{r}{j}(\cstate{j}{t},\chan_t)\wedge \pfunc{\trans}{r}{j}(\cstate{j}{t}, \cstate{j}{t+1}, \datfun_t, \chan_t)\wedge\pred_t}\\[4pt]
\vee &

\neg\pfunc{g}{r}{j}(\cstate{j}{t},\chan_t)\wedge \cstate{j}{t}=\cstate{j}{t+1}\\[4pt]
\vee  &

\chan_t=\toall\wedge \neg \pred_t\wedge \cstate{j}{t}=\cstate{j}{t+1}

\end{array}}}
\end{array}
\]

\end{definition}

That is, we use the information in the observation to localize the
sender $k$ and to specify the channel, data values, and the
send predicate.

The following theorem states a full abstraction property~\cite{milner75}, namely that the CTS semantics of systems and their discrete counterpart define the same transition relation, under closed world assumption. That is, by considering the messages originating from the system under consideration as the only observations.
\begin{theorem}[Full abstraction]\label{thm:dlts=ds} Given a set of \rcp agents $\{A_{\id}\}_{\id}$, their discrete system representation, defined as $
S=\conf{\mathcal{V}\coma\rho\coma\theta}$, is semantically equivalent to the parallel composition of their CTS representation, defined as $\mathcal{T}=\|_{\id}\mathcal{T}(A_{\id})$, under closed world assumption. More precisely,

\begin{itemize}
\item for every assignment $s$ to system variables $\mathcal{V}$, it follows that: $\theta(s)$ iff  $s\in S_0$;
\item for all assignments $s$ and $s'$ to variables in $\mathcal{V}$
  and respectively in $\mathcal{V}'$ it follows that: $\rho(s,s')$ iff
  there exist assignment to data variables $\datfun$, a communication
  channel $\chan$, and an agent $i$ such that  
  $(s,(\datfun,i,\pi,!,\chan),s')\in R_{\mathcal{T}}$.

\end{itemize}

\end{theorem}
\begin{proof} We prove each statement separately. 
\begin{itemize}
\item For $k$ agents in the symbolic representation, $\theta$ characterises the set of system states $S'\subseteq\Pi_{\id}S_{\id}$ that satisfy the initial conditions of all agents, i.e., $\set{s~|~s=(s_{0},s_{1},\dots,s_{k})~\mbox{and}~s\models \band{\id}{}{\theta_{\id}}}$. Note that $(s_{0},s_{1},\dots,s_{k})\models \band{\id}{}{\theta_{\id}}$ iff $s_{0}\models\theta_0\wedge s_{1}\models\theta_1\wedge\dots\wedge s_{k}\models\theta_k$. By Def.~\ref{def:par} and Def.~\ref{def:agentsem} this is exactly the set of initial states $S_0$ in $\mathcal{T}=\|_{\id}\mathcal{T}(A_{\id})$; 
\item By Def.~\ref{def:sys}, we have that $\rho(s, s')$ evaluates to
  true if there exists a valuation $\datfun$ to $\sdat$ and a channel
  $\chan$ in $\schan$ 
  such that both of the following hold:
 \begin{itemize}[label={$\bullet$}, topsep=0pt, itemsep=0pt, leftmargin=10pt]
 \item
   There exists an agent $\id$ such that the send transition
   $\pfunc{\trans}{s}{\id}$ is satisfied by assigning to current
   \emph{local} state $s_i$, next local state $s'_i$ (i.e., the
   projection of the system states $s$ and $s'$ on agent $\id$), the
   valuation $\datfun$, and the communication channel $\chan$.
   According to the enumerative semantics in Def.~\ref{def:agentsem},
   agent $\id$ has an individual send transition given the
   current \emph{local} state $s_{\id}$, next local state $s'_{\id}$,
   valuation $\datfun$ to data variables, and 
   channel $\chan$.
   Namely, agent $\id$ has a send transition
   $(s_i,(\datfun,i,\pi,!,\chan),s'_i)\in R_{\id}$ such that  
   $\pred=\pfunc{g}{s}{\id}(s_{\id},\chan, \datfun)$, i.e.,
   $\pred$ is obtained from $\pfunc{g}{s}{\id}$ by assigning the
   state of $\id$, the data variables assignment $\datfun$ and the channel $\chan$, and
   $\pfunc{\trans}{s}{\id}(s_{\id},s'_{\id},\datfun,\chan)$ evaluates to
   $\true$.
 \item
    For every other agent $\id'$ we have that either:
    \begin{enumerate}
    \item it is connected (i.e., $\pfunc{g}{r}{\id'}(s_{\id'},\chan)$
      holds), satisfies the send predicate (i.e.,
      $\pred(f^{-1}_{\id'}(s_{\id'}))$ holds), and participates in the 
      interaction (i.e.,
      $\pfunc{\trans}{r}{\id'}(s_{\id'},s'_{\id'},\datfun,\chan)$
      holds).
      By Def.~\ref{def:agentsem}, we know that agent $\id'$ has an
      individual receive transition
      $(s_{\id'},(\datfun,\id,\pi,?,\chan),s'_{\id'})\in R_{\id'}$
      where $\chan\in\listen^{\id'}(s_{\id'})$; 
     
     \item it is not connected (i.e.,
       $\neg\pfunc{g}{r}{\id'}(s_{\id'},\chan)$ ) and
       $s_{\id'}=s'_{\id'}$. By Def.~\ref{def:agentsem}, agent $\id'$
       does not have a receive transition for this message. In other
       words, since $\chan\notin\listen^{\id'}(s_{\id'})$ then agent
       $\id'$ cannot observe this transmission 
    ;
     
     \item or the message is sent on a broadcast channel
       ($\chan=\star$) where agent $\id'$ does not satisfy the sender
       predicate (i.e., $\neg\pred(f^{-1}_{\id'}(s_{\id'}))$)  
       and $s_{\id'}=s'_{\id'}$. By Def.~\ref{def:agentsem} this is
       equivalent to ignoring this message by not implementing a
       receive transition. 
    \end{enumerate}
    So far, we have shown that every individual (send/receive
    transition) in the symbolic model has a corresponding one in the
    enumerative semantics of individual agents. It remains to show
    that the composition of these individual transitions according to
    $\rho$ in the symbolic model has exactly the same semantics of the
    parallel composition in Def.~\ref{def:par}. That is, $\rho(s, s')$
    iff
    for the identified $\datfun$, $i$, $\chan$ and $\pred$ we have 
    $(s,(\datfun,i,\pi,!,\chan),s')\in R_{\mathcal{T}}$, given
    the assignments $s$ and $s'$ to variables in $\mathcal{V}$ and
    respectively in $\mathcal{V}'$.

    The existential
    quantification on sender transitions in $\rho$ (i.e.,\\
    $\bor{k}{}{\pfunc{\trans}{s}{k}(V_k, V'_k, \sdat, \chan)}$)
    implies that the order of the composition is immaterial, namely
    any two systems states $(s_0,s_1,\dots,s_k)$ and
    $(s_1,s_0,\dots,s_k)$ that only differ in the order of individual
    agent' states are semantically equivalent. By
    Lemma~\ref{lem:com+assoc}, we have that parallel composition is
    commutative, and thus the order is immaterial under the
    enumerative system semantics as well.
    If $\rho(s,s')$ is due to a message exchange on the  broadcast channel
    $\toall$ then the non-blocking semantics of the broadcast is
    preserved by the transition relation of the CTS composition as
    stated in Lemma~\ref{lem:brdnblk}. Moreover, if $\rho(s, s')$ is
    due to a message exchange on a multicast channel $c$ then the
    blocking semantics of the multicast is preserved by the transition
    relation of the CTS composition as stated in
    Lemma~\ref{lem:multi}.
    Lastly, the universal quantification on all
    possible receivers in $\rho$ (i.e., $\band{j\neq k}{}{}$) follows
    by the CTS semantics of parallel composition in Def.~\ref{def:par},
    where a receive transition can be received jointly by different
    agents, and by the commutativity and associativity of parallel
    composition (Lemma~\ref{lem:com+assoc}) where the scope of a send
    transition can be extended to cover all possible receivers.

    The other direction of the proof follows in similar manners. 
  \end{itemize}
  
  \end{itemize}
%
%
%

\end{proof}

The following is a corollary of Theorem~\ref{thm:dlts=ds} to relate the traces arising from Def.~\ref{def:systrace} to that of Def.~\ref{def:ts}.
\begin{corollary}[Trace equivalence]{\label{lem:dlts=ds}} The traces of a symbolic system 
composed of a set of agents $\{A_{\id}\}_{\id}$ 
  are the paths of the induced CTS.
\end{corollary}


\section{Reconfigurable Manufacturing Scenario}\label{sec:exp}
We complete the details of the RMS example, informally described in
Section~\ref{sec:overview}.
Many aspects of the example are kept simple on purpose to aid the presentation. 
 
The system, in our scenario, consists of an assembly product line agent 
(\val{line}) and several types of task-driven robots.
We describe the behaviour of the product line and only robots of
type-$1$ (\taval) as these are sufficient for exposing all features of
\rcp. 

A product line is responsible for assembling the main parts and delivering 
the final product.
Different types of robots are responsible for sub-tasks, e.g., retrieving and/or
assembling individual parts.
The product line is generic and can be used to produce different products 
and thus it has to determine the set of resources, to recruit a team of robots, 
to split tasks, and to coordinate the final stage of production. 

Every agent has copies of the common variables: $\typecvar$ indicating its type
(e.g., $\val{line}$, $\taval$, $\tbval$, $\tcval$),
$\assigncvar$ indicating whether a robot is assigned, and $\readycvar$
indicating what stage of production the robot is in.
The set of channels includes the broadcast channel $\toall$ and
multicast channels $\{\msf{A},\ldots \}$.
For simplicity, we only use the multicast channel $\msf{A}$ and fix it to the line agent.
The set of data variables includes $\msgdvar, \nodvar,$ and  $\lnkdvar$,
indicating the type of the message, a number (of robots per type),
and a name of a channel respectively.

We note that when a data variable is not important for some message it is omitted
from the description. 

We start with the description of the line agent \val{line}.
We give a high-level overview of the protocol applied by the line
agent using the state machine in Fig.~\ref{fig:line}.
The states capture a partial evaluation of the state variables of the
agent.
In this case, the value of the state variable $\stlvar$.
Transitions labels represent guarded commands.
We use the format \qt{${\bf\langle \msf{\Phi}
    \rangle\ \datfun\ \msf{\color{red}!/?}\ \msf{\chan[v'_1=a_1;\dots
        v'_n=a_n]}}$} to denote a guarded command $\msf{cmd}$.  
Namely, the predicate $\msf{\Phi}$ is a condition on the current
assignment to local variables of an agent (and for receive transitions
also on data variables that appear in the message).
We freely use $\datfun$ to refer to an assignment to data variables.
Usually, we write directly only the value of the $\msgdvar$ variable
to avoid cluttering. 
Sometimes, we add the values of additional data variables.
Each guarded command is labelled with a role (\modif{!} for send and
\modif{?} for receive transitions); also with a channel name
$\msf{\chan}$ and a new assignment to local variables
$\msf{[v'_1=a_1;\dots v'_n=a_n]}$ to represent the side effects of the
interaction.  
For the line agent, the protocol consists of starting from the pending
state and sending a team formation broadcast.
This is followed by sending of an assembly multicast on the channel
stored in local variable $\lnklvar$ and updating the stage to $2$.
Finally, an additional assembly multicast on the same channel resets
the process.
We include below the full description with the guards and predicates.
Each transition in the state machine corresponds to a disjunct in
either the send or the receive transition predicate below. 
Variables that are not assigned in a transition are kept unchanged in
the predicate. 
The send and receive guards of the agent are only partially captured
in the state machine.

\begin{figure}
\centering
\includegraphics[scale=.3]{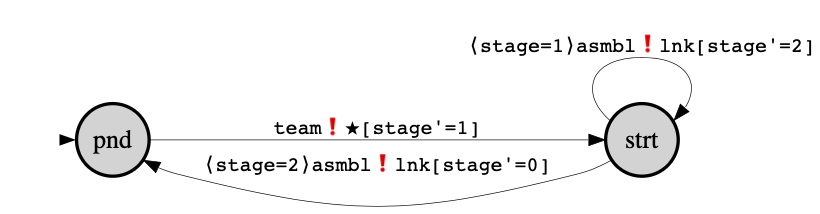}
\caption{Product Line Agent}
\label{fig:line}
\end{figure}

We now turn to the formal description of the line agent, starting with
its set of variables. 
In addition to copies of common variables (e.g., $f_l(\typecvar)$
$=\typelvar$), the  line agent has the following state variables:
$\stlvar$ is a state ranging over $\{\pendval, \startval\}$ (pending
and  start), $\lnklvar$ is the link of the product line, 
$\prdlvar$ is the id of the active product, and $\stagelvar$ is used
to range over the different stages of production.

The initial condition ${\theta_l}$ of a line agent is defined as follows:

\[
{\theta_{l}} : \stlvar=\pendval \wedge \stagelvar=\val{0}  \wedge \lnklvar=\msf{A} \wedge (\prdlvar=\val{1} \vee \prdlvar=\val{2})
\]

Thus, starting from the pending state, the line agent has a task of
assembling one of two products, and uses a multicast channel $\msf{A}$
to coordinate the assembly team. 
If there are multiple product lines, then each is initialised
with a dedicated channel. 

The send guard of the $\lineagent$ agent is of the following form:
\[
\begin{array}{l r r}
  \pfunc{g}{s}{l}  :  \chan{=}\star\wedge \neg\assigncvar \wedge 
 ( \prdlvar{=}\val{1}{\rightarrow}(\typecvar{=}\taval\vee \typecvar{=}\tbval))
    \wedge \\ 
    \qquad\qquad
        {(\prdlvar{=}\val{2}{\rightarrow}(\typecvar{=}\taval \vee \typecvar{=}\tcval))
        \vee
  \chan{=}\msf{lnk}\wedge \readycvar=\stagelvar}
\end{array}
\]

Namely, broadcasts are sent to robots whose $\assigncvar$ is false (i.e.,
free to join a team). If the identity of the product to be
assembled is 1, then the required agents are of types $\taval$ and
$\tbval$ and if the identity of the product is 2, then the required agents
are of types $\taval$ and $\tcval$. 
Messages on channel $\msf{A}$ (the value of $\lnklvar$) are sent to connected agents
when they reach a matching stage of production, i.e., $ \readycvar=\stagelvar$.
The receive guard of $\lineagent$ is $\chan=\star$, i.e., it is only connected to channel $\star$.

We may now proceed by explaining \rcp's send and receive transition
relations of the line agent in light of the state machine in
Fig.\ref{fig:line}. 
The send transition relation of $\lineagent$ is of the following form:
\[
\begin{array}{ll} 
  {\pfunc{\trans}{s}{l}}:\ 
  \keep(\lnklvar,\prdlvar,\typelvar,\assignlvar,\readylvar) \wedge \\
 \qquad\qquad\qquad\qquad
  {
    \left (
    \begin{array}{l l}
&  \stlvar=\pendval \wedge 
  \datfun(\msgdvar\mapsto \teamval;\nodvar\mapsto \val{2};\lnkdvar\mapsto\lnklvar)\\
\multicolumn{2}{r}{\wedge~ \stagelvar'=\val{1} \wedge
    \stlvar'=\startval \wedge \chan=\toall}\\[2ex]
\vee& \stlvar=\startval \wedge
  \datfun(\msgdvar\mapsto\assembleval) \wedge  \stagelvar=\val{1}\wedge\\
\multicolumn{2}{r}{\wedge~ \stlvar'=\startval \wedge \stagelvar'=\val{2}\wedge \chan=\lnklvar} \\[2ex]
\vee& \stlvar=\startval \wedge
    \datfun(\msgdvar\mapsto\assembleval) \wedge \stlvar'=\pendval \\
    \multicolumn{2}{r}{\wedge~ \stagelvar=\val{2}\wedge \stagelvar'=\val{0}\wedge \chan=\lnklvar}
    \end{array} \right ) }
\end{array}
\]
The $\lineagent$ agent starts in the pending state (see $\theta_l$).
It broadcasts a request ($\datfun(\msgdvar\mapsto\teamval)$) for two
robots ($\datfun(\nodvar\mapsto\val{2})$) per required type asking them to join
the team on the multicast channel stored in its $\lnklvar$ variable
($\datfun(\lnkdvar\mapsto\lnklvar)$).
According to the send guard, described before, if the identity of the product to assemble is 
1 ($\prdlvar=\val{1}$) the broadcast goes to type 1 and type 2 robots and if
the identity is 2 then it goes to type 1 and type 3 robots.
Thanks to channel mobility (i.e., $\msf{\datfun(\rulename{lnk})=lnk}$) a
team on a dedicated link can be formed incrementally at run-time. As a
side effects of broadcasting the $\teamval$ message, the line agent
moves to the start state $(\stlvar'=\startval)$ where the first stage
of production begins $(\stagelvar'=\val{1})$. 
In the start state, the line agent attempts an \rulename{assemble}
(blocking) multicast on $\msf{A}$.
The multicast can be sent only when the entire team completed the work
on the production stage (when their common
variable $\readycvar$ agrees with $\stagelvar$ as specified in the send guard).
One multicast increases the value of $\stagelvar$ and keeps
$\lineagent$ in the start state.
A second multicast finalises the production and $\lineagent$
becomes free again.

We set 
$
  \pfunc{\trans}{r}{l}\hspace{-1mm}:\keep(\msf{all})$
  as $\lineagent$'s recieve transition relation.
That is, $\lineagent$ is not influenced by incoming messages.

We now specify the behaviour of $\msf{\rulename{t1}}$-robots and show
how an autonomous and incremental one-by-one team formation is done
anonymously at run-time.
As before, we give a high-level overview of the protocol using the
state machine in Fig.~\ref{fig:robot}.
The team formation starts when unassigned robots are in pending states
($\pendval$).
From this state they may only receive a team message from a line
agent.
The message contains the number of required robots 
$\datfun(\nodvar)$ and a team link $\datfun(\lnkdvar)$.
The robots copy these values to their local variables
(i.e., $\lnklvar'=\datfun(\lnkdvar)$ etc.) and 
move to the start state ($\startval$).
From the start state there are three possible transitions:
\begin{itemize}[label={$\bullet$}, topsep=0pt, itemsep=0pt, leftmargin=10pt]
\item
  Join - move to state $\enval$ - a robot 
  joins the team by \emph{broadcasting} a $\formval$ 
  message to $\msf{\rulename{t1}}$-robots forwarding the
  number of still required robots
  ($\datfun(\rulename{no})=(\nolvar-\val{1})$) and the team link
  ($\msf{\datfun(\rulename{lnk})=lnk}$).
  This message is sent only if $\msf{no\geq\val{1}}$, i.e, at least one 
  robot is needed.
  From state
  $(\enval)$ the robot starts its mission.
\item
  Wait - stay in state $\startval$ - a robot \emph{receives} a $\formval$ message from a 
  robot, updating the number of still required robots (i.e., if 
  ${\datfun(\nodvar)>\val{0}}$).
\item
  Step back - return to state $\pendval$ - a robot \emph{receives} a
  $\formval$ message from a robot, informing that no more robots are
  needed, i.e., ${\datfun(\nodvar)=\val{0}}$.
  The robot disconnects from the team link, i.e., $\lnklvar'=\val{$\bot$}$.
  Thus it may not block interaction on the team link. 
\end{itemize}
After joining the team, a robot in state $\enval$ (i.e., with
$\stagebvar=\val{1}$) starts its mission
independently until it finishes
($\stagebvar'=\val{n}\wedge\readybvar'=1$).
We have used ($\dots$) to abstract the individual behaviour of the
robot in state ($\enval$).
In fact, each local step corresponds to a broadcast message
($\msf{local}$) that is hidden from other agents. This will be
clarified later in the send guard of
the robot which evaluates to false when ($\msf{local}$) is enabled. 

When all team robots finish their individual tasks (i.e., circled in
the self-loop on state $\enval$ while $\readybvar=1$ until $\stagebvar=n$), they become
ready to receive an $\assembleval$ message on $\msf{A}$, to start the
next stage of production (i.e, $\readybvar'=2$) while still staying in
$\enval$ state.

From this final stage (i.e.,
$\readybvar=\val{2}$) the robots are ready to receive the final
$\assembleval$ message to finalise the product and subsequently they
reset to their initial conditions.

As before, each transition corresponds to a disjunct in the send and
receive transition relations, which are fully specified later in this
section.

\begin{figure}
\centering
\includegraphics[scale=.35]{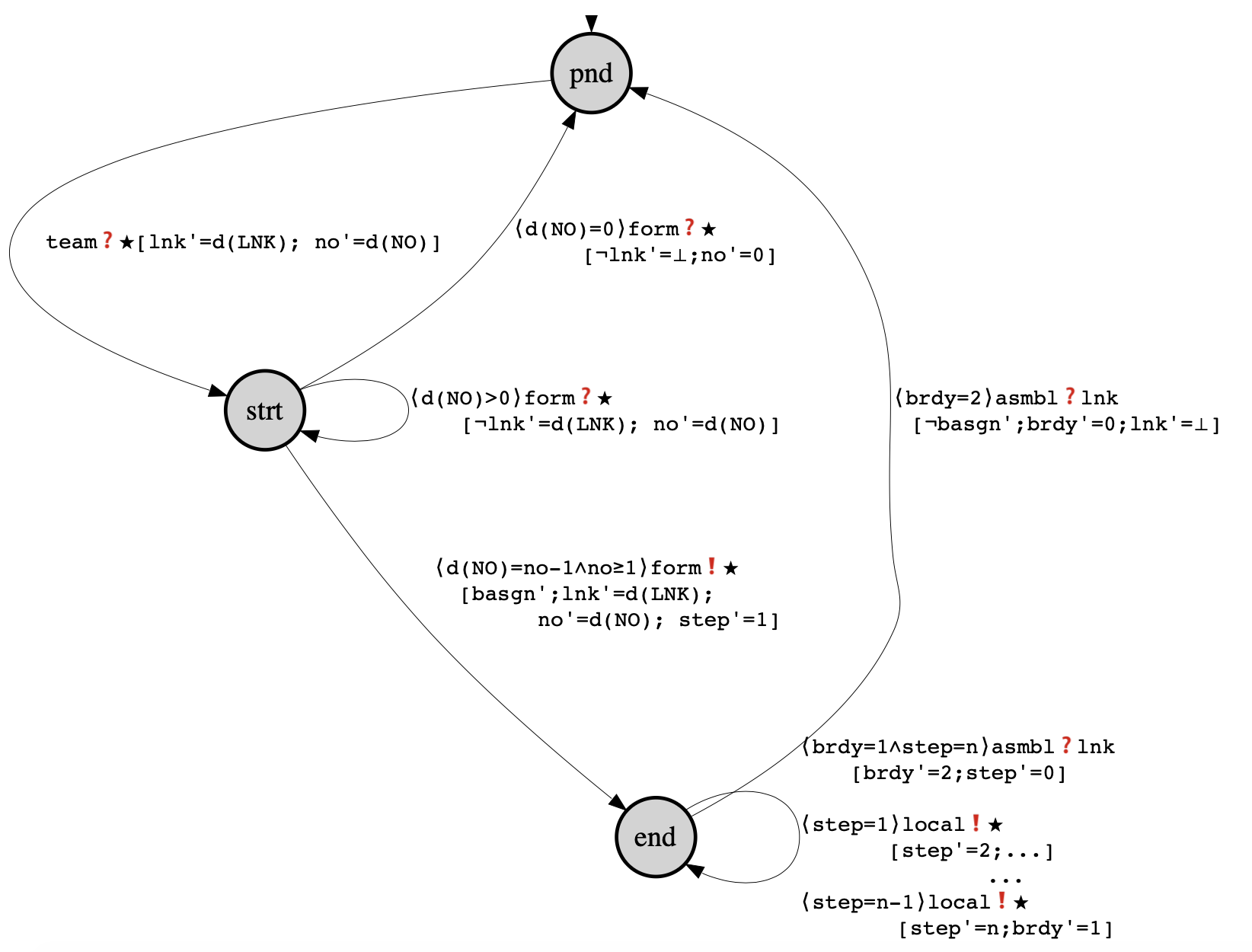}
\caption{The Agent of $\msf{\rulename{t1}}$-Robot}
\label{fig:robot}
\end{figure}

We now turn to the formal description of the robot, starting with its
set of variables. 
In addition to copies of common variables a $\msf{\rulename{t1}}$-robot has the following variables:
 $\stlvar$ ranges over $\set{\msf{\pendval, \startval, \val{end}}}$,
$\stagebvar$ is used to control the progress of individual behaviour,
$\msf{no}$ (resp. $\msf{lnk}$) is a placeholder to a number (resp. link) 
learned at run-time, and $f_b$ relabels common variables as follows: $f_b(\typecvar)=\typebvar$, $f_b(\assigncvar)=\assignbvar$ and $f_b(\readycvar)=\readybvar$.

Initially, a robot is in the pending state and is available for recruitment:

\[
\begin{array}{rcl}
{\theta_{b}} : \msf{(st=\pendval)\wedge(btype=\taval)\wedge \neg basgn\wedge(lnk=\bot)\wedge}\\
(\stagebvar=\readybvar=\nolvar=\val{0})
\end{array}
\]

The send guard of the robot is of the following form: 

\[
\begin{array}{lr}
\pfunc{g}{s}{b} : (\chan=\star)\wedge\datfun(\rulename{msg}\neq\msf{local})\wedge
  (\typecvar=\typebvar)\wedge \neg \assigncvar\ \vee\\ 
  \qquad\qquad\qquad\qquad(\chan=\star)\wedge\datfun(\rulename{msg}=\msf{local})\wedge (\assigncvar\wedge \neg \assigncvar)
\end{array}
\]

Interestingly, the send guard delimits the scope of the broadcast,
depending on the assignment to data variables.
Namely, it specifies that a robot either broadcasts to unassigned
robots of the same type if the message is not a local one
($\datfun(\rulename{msg}\neq\msf{local}$) or otherwise hides the
message from all other agents by broadcasting on a false predicate
(i.e., the tautology $\assigncvar\wedge \neg \assigncvar$). Note that
 such message
cannot be received by any agent, and it can be regarded as a local
computation.
Thus, it becomes very easy to distinguish the individual behaviour of
an agent from its interactions with the rest of the system.

The receive guard specifies that a $\msf{\rulename{t1}}$-robot is connected  either to
a broadcast $\star$ or to a channel matching the value of its link variable:
\[
\begin{array}{rcl}
\pfunc{g}{r}{b} : \chan=\star \vee \chan=\msf{lnk}.
\end{array}
\]

Finally, we report the send $\pfunc{\trans}{s}{b}$ and receive
$\pfunc{\trans}{r}{b}$ transition predicates below. 

\[
\begin{array}{ll} 
  {\pfunc{\trans}{s}{b}}:
  \keep(\lnklvar,\typebvar) \wedge \\
  \qquad\qquad\qquad\qquad
  {
    \tuple{
    \begin{array}{l l}
& 
\stlvar=\startval \wedge \datfun(\msgdvar\mapsto \formval;\lnkdvar\mapsto\lnklvar;\nodvar\mapsto\nolvar-\val{1)}\\
\multicolumn{2}{r}{\wedge~ (\nolvar\geq \val{1)}\wedge\stagebvar=\val{0}\wedge\stagebvar'=\val{1}\wedge\stlvar'=\enval
    }\\
  \multicolumn{2}{r}{\wedge \keep(\readybvar) \wedge \assignbvar'\wedge (\nolvar'=\val{0})\wedge
 \chan=\toall
  }\\[2ex]
  
  \vee& \stlvar=\stlvar'=\enval \wedge
  \datfun(\msgdvar\mapsto\localval) \wedge \chan=\toall \wedge\\
\multicolumn{2}{r}{
    \stagebvar=\val{1}\wedge\stagebvar'=\val{2}\wedge \keep(\assignbvar,\nolvar,\readybvar)}\\
&\vdots\qquad[\msf{\text{\scriptsize\rulename{individual behavior}}}]\\[2pt]
\vee& \stlvar=\stlvar'=\enval \wedge\datfun(\msgdvar\mapsto\localval)
\wedge\chan=\toall \wedge\stagebvar=\val{n-1}\\
	\multicolumn{2}{r}{
    \wedge\stagebvar'=\val{n}\wedge\readybvar'=\val{1} \wedge \keep(\assignbvar,\nolvar)}

    \end{array} }}
\end{array}
\]

\[
\begin{array}{ll} 
  {\pfunc{\trans}{r}{b}}:
  \keep(\typebvar) \wedge \\
  \qquad\qquad
  {
    \tuple{
    \begin{array}{l l}
& 
\stlvar=\pendval \wedge \datfun(\msgdvar\mapsto
\teamval)\wedge\stlvar'=\startval \wedge  \chan=\toall \wedge \\
\multicolumn{2}{r}{\wedge~ \lnklvar'=\datfun(\lnkdvar)\wedge\nolvar'=\datfun(\nodvar)\wedge
  \keep({\assignbvar},{\readybvar},{\stagebvar})
}\\[2ex]
  
\vee& \stlvar=\stlvar'=\startval \wedge \datfun(\msgdvar\mapsto
\formval)\wedge \datfun(\nodvar)>\val{0}\wedge\chan=\toall \wedge\\
\multicolumn{2}{r}{  \keep({\assignbvar},{\readybvar},{\stagebvar})\wedge\lnklvar'=\datfun(\lnkdvar)\wedge\nolvar'=\datfun(\nodvar)
 
    }\\[2ex]
\vee& \stlvar=\startval \wedge \datfun(\msgdvar\mapsto \formval;\nodvar\mapsto\val{0})\wedge
 \chan=\toall \wedge \stlvar'=\pendval\wedge\\
\multicolumn{2}{r}{\wedge   \keep({\assignbvar},{\readybvar},{\stagebvar})\wedge \lnklvar'=\bot\wedge\nolvar'=\val{0}
    }\\[2ex]
    
    \vee& \stlvar=\enval \wedge \datfun(\msgdvar\mapsto \assembleval)\wedge
 \readybvar=\val{1}\wedge\chan=\lnklvar \wedge \stagebvar=\val{n}
 \wedge \\
\multicolumn{2}{r}{\wedge~ \keep(\assignbvar,\lnklvar) \wedge\stlvar'=\enval\wedge\readybvar'=\val{2}\wedge\stagebvar'=\val{0}
    }\\[2ex]
    \vee& \stlvar=\enval \wedge \datfun(\msgdvar\mapsto \assembleval)\wedge
 \readybvar=\val{2}\wedge\chan=\lnklvar\\
\multicolumn{2}{r}{\wedge~ \stlvar'=\pendval\wedge\readybvar'=\val{0}
    \wedge \lnklvar'=\bot\wedge\neg\assignbvar'
    }
    \end{array} }}
\end{array}
\]

\section{\textsc{LTOL}: An extension of \textsc{LTL} }\label{sec:logic}
We introduce \ltal, an extension of \ltl with the ability to refer and therefore 
	reason about agents interactions.
  We replace the next operator of \ltl with the observation descriptors:
  \emph{possible} $\nxt{O}$ and \emph{necessary} $\alws{O}$, to refer to messages and the intended set of receivers.
  The syntax of formulas $\phi$ and \emph{observation descriptors} $O$ is as follows:
\smallskip

\[
\begin{array}{@{}l@{\ }r@{\,\,}c@{\,\,}l@{}}
&
\phi &::=& v \mid
  	\neg v \mid
  	\phi \lor \phi \mid
  	\phi \wedge \phi \mid
  	\phi \Until \phi \mid  
  	\phi \Release \phi \mid
  	\nxt{O} \phi \mid
  	\alws{O} \phi\\[2pt]
&
O & ::= &  
\cv \mid \neg \cv \mid \chan \mid \neg \chan \mid  k \mid \neg k \mid 
  	\dat \mid  \neg \dat\mid \exis{O} \mid \all{O} \mid O \lor O 
	  \mid O \wedge O
\end{array}
\]

We use the classic abbreviations $\Impl,\Iff$ and 
the usual definitions for $\true$ and $\false$. We also introduce
the temporal abbreviations $\f\phi\equiv 
\true\Until\phi$ (\emph{eventually}), $\g\phi\equiv\neg\f\neg\phi$
(globally) and $\varphi\WaitFor\psi\equiv\psi\Release
(\psi\vee\varphi)$ (\emph{weak until}). Furthermore we assume that
all variables are Boolean because  
every finite domain can be encoded by multiple Boolean
variables. For convenience we will, however, use non-Boolean
variables when relating to our RMS example.  

The syntax of \ltal is presented in \emph{positive normal form} to
facilitate translation into alternating B\"uchi automata (ABW) as
shown later. That is, we push the negation down to atomic propositions. We, therefore, use $\overline{\Theta}$ to denote the
dual of formula  
$\Theta$ where $\Theta$ ranges over either $\phi$ or $O$.
Intuitively, $\overline{\Theta}$ is obtained from $\Theta$ by switching 
$\vee$ and $\wedge$ and by applying dual to sub formulas, e.g.,
$\overline{\phi_1 \Until \phi_2} = \overline{\phi_1} \Release 
\overline{\phi_2}$, $\overline{\phi_1 \wedge \phi_2} = \overline{\phi_1} \vee 
\overline{\phi_2},\ $ $\overline{\cv} = \neg \cv$, and $\overline{\exis{O}} = 
\all{\overline{O}}$. 

Observation descriptors are built from referring to the
different parts  
of the observations and their Boolean combinations. Thus, they
refer to the channel in $\schan$, the data  
variables in $\sdat$, the sender $k$, and the predicate over common variables 
in $\scv$.
These predicates are interpreted as sets of possible assignments to
common variables, and therefore we include existential $\exis{O}$ and
universal $\all{O}$ quantifiers over these assignments.

The semantics of an observation descriptor $O$ is defined for an
observation
$m = \tuple{\chan\coma \datfun\coma k\coma \pred}$ as follows:

\[
\begin{array}{l l l @{\qquad\qquad } | @{\qquad\qquad }l l l}
 m \models \chan'& \text{\bf iff} & \chan = \chan'\quad 
 &\quad 

m \models \neg \chan' & \text{\bf iff} & \chan \neq \chan'\\

 m \models \dat' & \text{\bf iff} &\datfun(\dat') &\quad 

 m \models \neg \dat' & \text{\bf iff} & \neg \datfun(\dat') \\

m \models k' & \text{\bf iff} & k = k'&\quad 
m \models \neg k' & \text{\bf iff} & k \neq k' 
\end{array}
\]
\[
\begin{array}{lcc}
			m \models \cv\quad \textbf{iff}\quad \text{for
                          all}\ c \in \pred\ \text{we have}\ 
			c \models \cv&\\

			m \models \neg\cv \ \ \textbf{iff}\ \ \text{there is}\ c \in \pred\ \text{such that}\ 
			c \not\models \cv&\\

			m \models \exis{O} \ \ \textbf{iff}\ \ \text{there is}\ c \in \pred \ \text{such that}\ 

			\tuple{\chan, d, k, \{c\}} \models O&\\

			m \models \all{O} \ \ \textbf{iff}\ \ \text{for all}\ c \in \pred\ \text{it holds that}\ 
			\tuple{\chan, d, k, \{c\}} \models O&\\
		
			m \models O_1 \vee O_2 \quad \textbf{iff}\quad \text{either}\ m \models O_1\ \text{or}\ m 
			\models O_2&\\

			m \models O_1 \wedge O_2 \quad \textbf{iff}\quad  m \models O_1\ \text{and}\ m 
			\models O_2&
	\end{array}
	\]

We only comment on the semantics of the descriptors $\exis{O}$ and
$\all{O}$ as the rest are standard propositional formulas. 
The descriptor $\exis{O}$ requires that at least one
assignment $c$ to the common variables in the sender predicate
$\pred$ satisfies $O$.
Dually $\all{O}$ requires that all assignments in $\pred$ satisfy
$O$.
Using the former, we express properties where we require that the
sender predicate has a possibility to satisfy $O$ while using the
latter we express properties where the sender predicate can only
satisfy $O$.
For instance, both observations $\tuple{\chan, \datfun, k, \cv_1\vee\neg\cv_2}$ 
and $\tuple{\chan, \datfun, k, \cv_1}$
satisfy $\exis{\cv_1}$ while only the latter satisfies
$\all{\cv_1}$.
Furthermore, the observation descriptor $\all{\false}\wedge\chan=\star$ says
that a message is sent on the broadcast channel with a false
predicate.
That is, the message cannot be received by other agents.
In our RMS example in Section~\ref{sec:exp}, the
descriptor $\exis{(\typecvar=\taval)} \wedge 
\all{(\typecvar=\taval)}$ says that the message is intended exactly for
robots of type-$1$.

Note that the semantics of $\exis{O}$ and $\all{O}$ (when nested) ensures that 
the outermost cancels the inner ones, e.g., $\exis{(O_1\vee 
  (\all{(\exis{O_2})}))}$ is equivalent to $\exis{(O_1\vee O_2)}$.
Furthermore, when $\cv$ and respectively $\neg\cv$ appear outside the
scope of a quantifier ($\all{}$ or $\exis{}$), they are semantically
equivalent to the descriptors $\all{\cv}$ and respectively
$\exis{\neg\cv}$. 
Thus, we assume that they are written in the latter normal 
form.

We interpret \ltal formulas over system computations:
\begin{definition}[System computation]~\label{def:syscomp}
A system computation $\rho$ is a function from natural numbers $N$ to
$\Exp{\sysvar}\times M$ where $\sysvar$ is the set of state variable
propositions and $M=\schan\times 2^{\sdat}\times K\times\dexp{\scv}$
is the set of possible observations. That is, $\rho$ includes values
for the variables in $\Exp{\sysvar}$ and an observation in $M$ at each
time instant.
\end{definition}

We denote by $\sstate{\id}$ the system state at the $i$-th 
time point of the system computation.
Moreover, we denote the suffix of $\tracevar{}$ starting with the $i$-th
state by $\tracevar{\geq i}$ and we use $m_{i}$ to denote the observation
$\tuple{\chan, \datfun, k, \pred}$ in $\tracevar{}$ at time point
$i$.

The semantics of an \ltal formula $\varphi$ is defined for a computation 
$\tracevar{}$ at a time point $\id$ as follows:\smallskip

	\[\begin{array}{lcc}		
			\tracevar{\geq i}\models v\ \ \textbf{iff}\ \ \sstate{\id} \models v\quad
		\text{and} \quad
			 \tracevar{\geq i}\models \neg v\ \ \textbf{iff}\ \  \sstate{\id} 
			 \not\models v;\\		
	
			\tracevar{\geq i}\models \phi_2\vee\phi_2\ \ \textbf{iff}\ \ \tracevar{\geq 
			i}\models\phi_1\ \text{ or }\ \tracevar{\geq 
			i}\models\phi_2;\\

			\tracevar{\geq i}\models \phi_2\wedge\phi_2\ \ \textbf{iff}\ \ 
			\tracevar{\geq i}\models\phi_1\ \text{ and }\ \tracevar{\geq 
			i}\models\phi_2;\\

			\tracevar{\geq i}\models \phi_1\Until\phi_2\ \ \textbf{iff}\ \text{there exists}\ 
			j\geq i\ \text{s.t.}\ \tracevar{ \geq j} \models \phi_2\ \text{and,}\\ \hfill\text{for 
			every}\ i\leq k<j,\ \tracevar{\geq k}\models\phi_1;\qquad\\

			\tracevar{\geq i}\models \phi_1\Release\phi_2\ \ \textbf{iff}\ \text{for every}\ j 
			\geq i \ \text{either}\  \tracevar{\geq j} \models \phi_2\ \text{or},\\ \hfill\text{there exists}\ i \leq 
			k<j, \tracevar{\geq k}\models\phi_1;\quad \\

			\tracevar{\geq i}\models\nxt{O}\phi\ \ \textbf{iff}\ \ m_i\models O\ 
			\text{ and }\  \tracevar{\geq i+1}\models\phi;\\			
	
			\tracevar{\geq i}\models\alws{O}\phi\ \ \textbf{iff}\ m_i\models O\ 
			\text{ implies }\  \tracevar{\geq i+1}\models\phi.\\
\end{array}\]\medskip

Intuitively, the temporal formula $\nxt{O}\phi$ is satisfied on the computation 
$\tracevar{}$ at point $\id$ if the observation $m_i$ satisfies $O$ \emph{and}
$\phi$ is satisfied on the suffix computation $\tracevar{\geq \id+1}$.
On the other hand, the formula $\alws{O}\phi$ is satisfied on the computation 
$\tracevar{}$ at point $\id$ if $m_i$ satisfying $O$ \emph{implies}
that $\phi$ is satisfied on the suffix computation $\tracevar{\geq \id+1}$.
Other formulas are interpreted exactly as in \ltl.

With observation descriptors we can refer to the intention of agents in the interaction.
For example, the following descriptor 
\[O:=\exis{(\typecvar=\taval)}\wedge\exis{(\typecvar=\tbval)}\wedge
\all{(\typecvar=\taval \vee \typecvar=\tbval)}\] specifies that the target of the
message is \qt{exactly and only} type-1 and type-2 robots.
This descriptor can be used later to specify that whenever the line agent \qt{$l$} recruits for a product with identity 1, 
 it notifies both type-1 and type-2 robots as follows:
\[\g((\prdlvar=\val{1} \wedge \stlvar=\pendval \wedge \nxt{l\wedge \chan=\toall}\true)
\rightarrow \nxt{O}\true)\]
 Namely, whenever the line agent is in the
pending state and tasked with product $1$ it notifies both type-1 and
type-2 robots by a broadcast. 
The pattern ``After $q$ have exactly two $p$ until $r$''
\cite{MR15,DAC99} can be easily expressed in \ltl and can be used to
check the formation protocol. 
Consider the following formulas:
\[
\varphi_1:=\langle \msgdvar=\teamval\wedge \nodvar=\val{2} \wedge
  \exis{(\typecvar=\taval)}\rangle\true\] specifying that a team message is sent to
type-1 robots and requires two robots,
\[\varphi_2:=\nxt{\msgdvar=\formval \wedge \exis{(\typecvar=\taval
    )}}\true\] specifying that a formation message is sent to type-1
robots, and
    
\[\varphi_3:=\nxt{\chan=\msf{A}}\true\] specifying that a message is sent on
channel $\msf{A}$.

Now, the template
``After $\varphi_1$ have exactly two $\varphi_2$ until $\varphi_3$''
specifies that whenever a team message is sent to robots of type-1
requiring two robots, then two form messages destined for type-1
robots will follow before using the multicast channel.
That is, two type-1 robots join the team before a (blocking) multicast
on channel $\msf{A}$ may become possible.

We can also reason at a
local rather than a global level.
For instance, we can specify that robots follow a ``correct''
utilisation of channel $\msf{A}$. Formally,
\[O_1(t):=\msgdvar{=}\teamval \wedge
  \exis(\typecvar{=}\val{t})\] specifies that a team message is sent
  to robots of type $\val{t}$;
  \[O_2(k,t):=\msgdvar{=}\formval \wedge
  \neg k \wedge
  \nodvar{=}\val{0} \wedge \exis{(\neg \assigncvar \wedge
    \typecvar{=}\val{t})}\] specifies that a robot different from $k$ sends a
  form message specifying that no more robots are needed and this message is
  sent to unassigned type $\val{t}$ robots;
  \[O_3(t):=\msgdvar{=}\assembleval \wedge \chan=\msf{A} \wedge
    \readycvar{=}\val{2} \wedge \exis{(\typecvar{=}\val{t})}\] specifies that an
    assembly message is sent on channel $\msf{A}$ to robots of type
    $\val{t}$ who reached stage 2 of the production.
    Thus, for robot $k$ of type $\val{t}$, the formulas

    \begin{equation}\label{eq:local}
    \begin{array}{l@{\ \ }l}
    \rom{1} & \varphi_1(t):=(\lnklvar{\neq} \msf{A}) \WaitFor \nxt{O_1(t)}\true\\
      \rom{2}&\varphi_2(k,t):=\g(\alws{O_2(k,t)\vee O_3(t)}\varphi_1(t))
    \end{array}
   \end{equation}

    state that: \rom{1} robots are not connected to channel $\msf{A}$ until they get
a team message, inviting them to join a team; \rom{2} if either they
are not selected ($O_2(k,t)$) or they finished production after selection
 ($O_3(t)$) then they disconnect again until the next team message.
This reduces to checking the  ``correct'' utilisation of channel $\msf{A}$ to individual
level, by verifying these properties on all types of robots
independently. 
%
%
By allowing the logic to relate to the set of targeted robots, verifying all
 targeted robots separately entails the
correct \qt{group usage} of channel $A$.

\subsection{The Satisfiability and the Model Checking problems of \ltal}
In this section, we improve our early results on satisfiability and model checking of \ltal, presented in the \rulename{aamas} version~\cite{rcp} of this article. In that version, we computed an \expspace upper bound  for both problems with respect to the set of common variables $\scv$ that appear in the observation descriptors and \pspace upper bound with respect to the rest of the input. This result was not surprising as the semantics of observations requires quantification on the assignments to common variables $\scv$ appearing in $O$. Indeed, the number of assignments to $\scv$ is doubly exponential in the size of $\scv$, i.e, the number of assignments is $\Exp{\Exp{\size{\scv}}}$. Both problems require translation to Nondeterministic B\"uchi Automata (NBW), and a direct translation would incur a double exponential blowup in the size of the automaton with respect to $\size{\scv}$. Thus, a membership in \expspace with respect to $\size{\scv}$ follows from the membership in \nlogspace of the nonemptiness problem for NBW. 

In this article, we improved the latter results to \pspace, matching
the lower bound. This is achieved by a novel automaton
construction. Namely, we introduce a further dependency between the
formula and the alphabet that is read by the automaton. Thus, the
automaton does not read concrete messages but it rather partitions
messages into sets, according to their effects on the truth values of
subformulas of the formula.
  
Before we proceed with the automaton construction, we fix the sets of system variables $\sysvar$, the communication channels $\schan$, the data
variables $\sdat$, the identities of agents $K$, and the common variables $\scv$. 

Our direct construction in~\cite{rcp} considers a \emph{state-alphabet}
$\Sigma=2^{\sysvar}$ and a \emph{message-alphabet} $M=\schan \times
{\sdat} \times K \times \dexp{\scv}$.
Clearly, the message-alphabet is doubly-exponential in $\scv$ and implies
that the decision procedures based on $M$ would be in
\expspace (with respect to $\scv$).
However, $M$ is ``too large'' for the automaton (c.f., \cite{Yan06}).
Thus, we consider a smaller alphabet that is derived from the
observation descriptors appearing in the formula.
This alphabet is at most exponential in the size of the formula
(allowing for \pspace analysis).
To achieve \pspace analysis, we have to extend the decision procedures to further  consider  \emph{observation-alphabet
satisfiability} and  \emph{observation-alphabet model-checking}, as we will see below.

Recall the alphabets $\Sigma$ and $M$ above and fix an \ltal formula
$\varphi$.
Let $\obs{\varphi}$ be the set of observations appearing ``top-level" in
the operators $\nxt{\cdot}$ and $\alws{\cdot}$ in $\varphi$.
More precisely, $\obs{\varphi}$ is closed under the subformula
relation of $\varphi$, but is not closed under the subformula relation
of $O$.
Consider $\varphi_2(k,t)$ in Equation~\ref{eq:local}: 
\[\obs{\varphi_2(k,t)}=\set{O_2(k,t)\vee O_3(t), O_1(t)}\]
We denote by $\size{\obs{\varphi}}$ the size of the set
$\obs{\varphi}$.
We denote by $\length{O}$ the length of the observation $O$ and by
$\length{\obs{\varphi}}$ the sum of lengths of observations in
$\varphi$. 
Note that $\length{\obs{\varphi}}$ is bounded by the size of
$\varphi$.
Thus, we may now define an \emph{observation-alphabet}
$\mcal{O}=2^{\obs{\varphi}}$, that is 
at most exponential in the size of $\varphi$. We will use this alphabet
to enable \pspace analysis. 

In our construction, the automaton reads words from the alphabet
$(\Sigma\times\mcal{O})^\omega$ while system computations are derived
from the alphabet $(\Sigma\times M)^\omega$. 

Intuitively, an automaton word $w\in(\Sigma\times\mcal{O})^\omega$ and
a system computation $\rho\in(\Sigma\times M)^\omega$ agree on the
state-alphabet $\Sigma$ and only differ in their treatment to
messages.
Formally, given a word $w=(\sigma_0,\smallo_0),(\sigma_1,\smallo_1),\dots$, and a system
computation 
$\rho = (\sigma'_0,m_0),(\sigma'_1,m_1),\dots$.
We say that $\rho$ satisfies $w$ if for every $i\geq 0$ we have that
$\sigma_i'=\sigma_i$ and for every $O\in \obs{\varphi}$ we have
$m_i\models O$ iff $O\in \smallo_i$.
Note that $m_i\models O$ follows the semantics of observation
descriptors.
Thus, a word $w$ defines a language over system computations. 

More precisely, for a word $w\in (\Sigma\times \mcal{O})^\omega$ we denote by
$\lang{w}$ the set of system computations satisfying $w$. 
We say that $w$ is \emph{non empty} if there is some system computation satisfying
it, i.e., if $\lang{w}\neq \emptyset$.
Furthermore, 
for a letter $\smallo\in\mcal{O}$, we denote by $\bcal{M}(\smallo)=\{m\in M
~|~ \forall O\in \obs{\varphi} ~.~ O\in \smallo \iff m\models O\}$
the set of models of $\smallo$.
That is, all the messages
that satisfy all the observations in $\smallo$ and do not satisfy all
the observations that are not in $\smallo$. 
We say that $\smallo$ is \emph{non empty} if $\bcal{M}(\smallo)\neq
\emptyset$.

Clearly, a word $w=(\sigma_0,\smallo_0),(\sigma_1,\smallo_1),\dots$ is non empty if and only if for every $i\geq 0$ we have that
$\smallo_i$ is non empty.

We show that satisfiability of \ltal can be reduced to finding a
word $w$ such that the set of system computations satisfying $w$ is not
empty.
Similarly, model checking is reduced to building an automaton for $\neg \varphi$ and identifying a word $w$ satisfying $\neg \varphi$ and
a computation $\rho$ of the system under study such that $\rho$
satisfies $w$.

The following theorem states that the set of computations satisfying a
given formula are exactly the ones satisfying words accepted
by some finite automaton on infinite words.

\begin{theorem} \label{thm:main}
For every \ltal formula $\varphi$, there is an Alternating B\"uchi Automaton (ABW) $A_{\varphi}=\langle Q,
  \Sigma, \mcal{O}, \transmain,$ $q_0, F\subseteq Q \rangle$ such that
$\bigcup_{w \in \lang{A_{\varphi}}} \lang{w}$ is exactly the set of
computations satisfying the formula $\varphi$. 
\end{theorem}

Notice that for a given word $w$, either \emph{all} the computations that
satisfy $w$ satisfy $\varphi$ or \emph{all} the computations that
satisfy $w$ do not satisfy $\varphi$ (i.e., satisfy
$\overline{\varphi}$).
In the first case $w$ is accepted by $A_{\varphi}$ and in the second it
is not accepted by $A_{\varphi}$.
Thus, the definition of $\mcal{O}$ is such that words partition the
computations to equivalence sets that are uniform with respect to the
satisfaction of $\varphi$. 

\begin{proof} 
The set of states $Q$ is the set of all sub formulas of $\varphi$ with $\varphi$ 
being the initial state $q_0$.
The automaton has two alphabets, namely the state-alphabet $\Sigma = 
2^{\sysvar}$ and the observation alphabet $\mcal{O}=
2^{\obs{\varphi}}$.
The set $F$ of accepting states consists of all sub formulas of the form 
$\phi_1 \Release \phi_2$.
The transition relation $\transmain: Q \times \Sigma \times \mcal{O} \rightarrow 
\B^{+}(Q)$ is defined inductively on the structure of $\varphi$, as follows:

\begin{itemize}[label={$\bullet$}, topsep=0pt, itemsep=0pt, leftmargin=10pt]
	\item
		$\transmain(v, \sigma, \smallo) = \true$ if $v\in\sigma$ and $\false$ otherwise;

	\item
		$\transmain(\neg v, \sigma, \smallo) = \true$ if $ v \not \in \sigma$ and 
		$\false$ otherwise;

	\item
		$\transmain(\phi_1 \vee \phi_2, \sigma, \smallo) = \transmain(\phi_1, \sigma, \smallo) 
		\vee \transmain(\phi_2, \sigma, \smallo)$;

	\item
		$\transmain(\phi_1 \wedge \phi_2, \sigma, \smallo) = \transmain(\phi_1, \sigma, 
		\smallo) \wedge \transmain(\phi_2, \sigma, \smallo)$;

	\item
		$\transmain(\phi_1 \Until \phi_2, \sigma, \smallo) = \transmain(\phi_1, \sigma, 
		\smallo) \wedge\phi_1 \Until \phi_2 \vee \transmain(\phi_2, \sigma, \smallo)$;

	\item
		$\transmain(\phi_1 \Release \phi_2, \sigma, \smallo) = (\transmain(\phi_1, 
		\sigma, \smallo) \vee \phi_1 \Release \phi_2) \wedge \transmain(\phi_2, \sigma, 
		\smallo)$;
	
	\item
	  $\transmain(\nxt{O}\phi_1, \sigma, \smallo) =
          \left \{ \begin{array}{l l} \phi_1 & O\in \smallo\\ \false &
            O\notin\smallo \end{array} \right .$;
		
	\item
	  $\transmain(\alws{O}\phi_1, \sigma, \smallo)=
          \left \{ \begin{array}{l l} \phi_1 & O \in \smallo \\ \true &
            O\notin \smallo \end{array} \right .$.
\end{itemize}

The proof of correctness of this construction proceeds by induction on
the structure of $\varphi$.

We prove that when $A_{\varphi}$ is in state $\phi_1$, it accepts exactly
all computations that satisfy $\phi_1$. The base cases (i.e., state
variable propositions) follow from the definition of $\transmain$
while other cases follow from the semantics of $\varphi$ and the
induction hypothesis. The construction ensures that a computation
can only satisfy $\phi_1\Until\phi_2$, if it has a suffix satisfying
$\phi_2$; otherwise $A_{\phi}$ will have an infinite path stuck in
$\phi_1\Until\phi_2$ which is not accepting. 
\end{proof}

Note that, from Theorem~\ref{thm:main}, the number of states in $A_{\varphi}$ 
is linear in the size of $\varphi$, i.e., $\size{Q}$ is in $\bigo{\size{\varphi}}$. 
The size of the transition relation $\size{\transmain}$ is in
$\bigo{\size{Q}^2.\size{\Sigma}.\size{\mcal{O}}}$, i.e.,
it is in ${\size{\varphi}^2.\Exp{\bigo{\size{\varphi}}}}$.
Finally, the size of
the alternating automaton $\size{A_{\varphi}}$ is
in $\bigo{\size{Q}.\size{\transmain}}$, i.e., $\size{A_{\varphi}}$ is
in $\size{\varphi}^3.\Exp{\bigo{\size{\varphi}}}$.

By Theorem~\ref{thm:main} and Proposition~\ref{prop:aton}, we have that:
\begin{corollary}
  For every formula $\varphi$ there is an NBW $N_\varphi$ with a
  state-alphabet $\Sigma=2^{\sysvar}$ and an observation-alphabet  
  $\mcal{O}=2^{\obs{\varphi}}$ where $N_\varphi=\langle Q, \Sigma, \mcal{O},
  S^0, \delta, F\rangle$ and $\bigcup_{w \in \lang{N_{\varphi}}}
  \lang{w}$ is exactly the set of computations satisfying $\varphi$ such
  that:  

\begin{itemize}[label={$\bullet$}, topsep=0pt, itemsep=0pt, leftmargin=10pt]
\item $\size{Q}$ is in $\Exp{\bigo{\size{\varphi}}}$
and 
  $\size{\delta}$ is in
$\bigo{\size{Q}^2.\size{\Sigma}.\size{\mcal{O}}}$, i.e.,
$\size{\delta}$ is in $\Exp{\bigo{\size{\varphi}}}$.

\item The required space for building the automaton is
  $\nlog{(\size{Q}.\size{\delta})}$, i.e., it is in
  $\bigo{\size{\varphi}}$ 

\item The size of the B\"uchi automaton is $\size{Q}.\size{\delta}$,
  i.e., $\size{N}$ is in $\Exp{\bigo{\size{\varphi}}}$.
\end{itemize}
\label{cor:1}
\end{corollary}
\begin{theorem}\label{thm:satisfiability}
The satisfiability problem of \ltal is \pspace-\rulename{complete} with respect to 
$\size{\varphi}$.
\end{theorem}

\begin{proof}
By Corollary~\ref{cor:1}, given a formula $\varphi$, we can construct an
NBW $N_\varphi$ of size $\size{Q_n}.\size{\delta_n}$ that
accepts precisely the computations that satisfy $\varphi$. Thus, $\varphi$
is satisfiable iff $N_\varphi$ is nonempty. In order to prove that the
formula is satisfiable we have to show that $N_\varphi$ accepts a word $w$
such that some computation $\rho$ satisfies $w$.
However, a word $w$ is non empty iff every letter
$\smallo\in\mcal{O}$ appearing in $w$ is non empty.
It follows that while testing the non emptiness of $N_\varphi$ we have to
follow only transitions using non empty letters in $\mcal{O}$.
The nonemptiness of an NBW is tested in nondeterministic logarithmic space.
However, as $N_\varphi$ is exponential in $\size{\varphi}$ we get an algorithm
working in space polynomial in $\size{\varphi}$.
The algorithm constructs
$N_\varphi$ on-the-fly. 
We have to show that the emptiness of letters in $\mcal{O}$ can be
tested in space polynomial in $\size{\varphi}$.
This follows from Proposition~\ref{prop:observation emptiness} below. 

The hardness argument can be proved
by a reduction from \ltl satisfiability~\cite{sistlac85}.
\end{proof}

\begin{theorem}\label{thm:modcheck}
The model-checking problem of \ltal is \pspace-\rulename{complete} with respect
to $\size{Sys}$ and $\size{\varphi}$.
\end{theorem}

Note that the stated bounds in terms of $\size{Sys}$ refer to the
symbolic representation of the system. 
The complexity is \rulename{logspace}
in the size of the corresponding CTS $\mathcal{T}(Sys)$, which is anyway exponentially larger.

\begin{proof}
Given a finite state system $Sys=\conf{\sysvar\coma\rho\coma\theta}$
and a set of assertions on state variables $\sysvar$, on
$\schan,\ {\sdat},\ K$, and on $cv_1,\dots cv_n$. We assume $\rho$ to
be total and then we can construct a CTS representation of $Sys$ as follows: $\mathcal{T}(Sys)=
\langle \schan,\Sigma,M,S,S_0,R,L,\rulename{ls}\rangle$, where the components
of $\mathcal{T}(Sys)$ are as follows.  
$S=\Sigma$ ($L$ is the identity function), and thus $S$ is the set of possible interpretations of the variables in
$\sysvar$, i.e.,
$S=\Exp{\sysvar}$. The set of initial states $S_0$ is the set of states $s$
such that $s\models \theta$, i.e., $S_0=\set{s\models {\theta}}$, and
$M=\schan\times\Exp{\sdat}\times K\times \dexp{\scv}$.
We have that
$R(s,m)=\set{s': \tuple{s,m,s'}\models \rho}$ and $\emptyset$
otherwise.
Furthermore, we consider all states in $\mathcal{T}(Sys)$ to be
accepting.
The number of states in the transition system $\mathcal{T}(Sys)$ may
be exponentially larger than the description of $Sys$.
Notice that although $M$ is doubly exponential in $\scv$ the labels of
transitions of $\mathcal{T}(Sys)$ are those obtained from
$\pfunc{\trans}{s}{k}$ for some $k$.
Thus, the number of distinct labels appearing on transitions of
$\mathcal{T}(Sys)$ is bounded by $|S|\cdot |\schan|\cdot
\Exp{|\sdat|}\cdot |K|$.

The system $Sys$ satisfies $\varphi$ iff all the computations of $Sys$ satisfy
$\varphi$, thus for every computation $\rho \in \lang{\mathcal{T}(Sys)}$
there exists a word $w\in \lang{N_{\varphi}}$ such that $\rho\models w$.
Dually, $Sys$ does not satisfy $\varphi$ iff for some computation $\rho
\in \lang{\mathcal{T}(Sys)}$ and for some word
$w\in \lang{N_{\neg \varphi}}$ we have $\rho\models w$.
%
%
This is equivalent to check
$\lang{\mathcal{T}(Sys)}\cap\bigcup_{w\in \lang{N_{\neg\varphi}}}\lang{w}=\emptyset$.
Since our formulas are in positive normal form, $\neg\varphi$ can be
obtained from $\varphi$ by $\overline{\varphi}$.
By Corollary~\ref{cor:1}, we have that
$N_{\neg\varphi}$ has $\Exp{\bigo{\size{\varphi}}}$ states and
$\size{N_{\neg\varphi}}$ is in
$\Exp{\bigo{\size{\varphi}}}$.
Note, however, that the words of $N_{\neg\varphi}$ are in
$(\Sigma\times\mcal{O})^{\omega}$ while the computations of
$\mathcal{T}(Sys)$ are in $(\Sigma\times M)^{\omega}$. 
The model-checking problem can be reduced to finding a word $w$
accepted by $A_{\neg\varphi}$ and a computation $\rho$ of
$\mathcal{T}(Sys)$ such that $\rho\models w$.
Recall that $\rho\models w$ if for every $i\geq 0$ we have that 
$\sigma^{\rho}_i=\sigma^{w}_i$ and for every $O\in \obs{\varphi}$ we
have $m^{\rho}_i\models O$ iff $O\in \smallo^{w}_i$.
This amounts to check the nonemptiness problem of
a (modified) intersection of $\mathcal{T}(Sys)$ and $N_{\neg\varphi}$,
where the transition $(s,m,s')$ of $\mathcal{T}(Sys)$ can match
transitions of $N_{\neg\varphi}$ that read letters $(s,\smallo)$ for
$m\models\smallo$.
Note that for every $m\in M$ there is a unique $\smallo\in\mcal{O}$
such that $m\models \smallo$.
Thus, we check letter by letter that the word $w$ accepted by
$N_{\neg\varphi}$ and the computation $\rho$ produced by
$\mathcal{T}(Sys)$ are such that $\rho\models w$.
Thus, we only need
to show that checking $m\models\smallo$ can be tested in space
polynomial in $\size{\varphi}$.
Indeed, 
This follows from Proposition~\ref{prop:observation checking}.
Since all states in $\mathcal{T}(Sys)$ are accepting, the construction
of $N_{\mathcal{T}(Sys),\neg\varphi}$ is the 
product of $\mathcal{T}(Sys)$ with $N_{\neg\varphi}$ with transitions
composed as explained.
We have that $N_{\mathcal{T}(Sys),\neg\varphi}$
has $\Exp{\bigo{\size{Sys}+\size{\varphi}}}$ states.
Hence, $\size{N_{\mathcal{T}(Sys), \neg\varphi}}$ is in
$\Exp{\bigo{\size{Sys}+\size{\varphi}}}$ .
We have that $N_{\mathcal{T}(Sys),\neg\varphi}$ can be constructed on-the-fly and a
membership in \pspace with respect to $\size{Sys}$ and $\size{\varphi}$,
follows from the membership in \nlogspace of the nonemptiness problem for
NBW.
Checking that $Sys\models\varphi$ is in $\bigo{{\size{\varphi}+\size{Sys}}}$. 

The hardness follows from the same hardness results for discrete systems
and \ltl~\cite{sistlac85}.
\end{proof}

The following proposition states that given a letter
$\smallo\in\mcal{O}$ we can check whether there exists a message $m$
that satisfies $\smallo$ in \rulename{np} with respect to
$\length{\obs{\varphi}}$.
Notice that, in particular, $\length{\obs{\varphi}}$ should be larger
than the number of variables in $\scv$ and $\sdat$ that
appear in $\varphi$, the number of agents in $K$ that are
mentioned in $\varphi$ and the number of channels in $\schan$
appearing in $\varphi$.
Those that do not appear in $\varphi$ can be removed from the message
alphabet $M$. 

\begin{proposition}[Observation satisfiability]
  Consider a letter $\smallo\in\mcal{O}$.
  Emptiness of $\smallo$ is \rulename{np-complete} in
  $\length{\obs{\varphi}}$.
  \label{prop:observation emptiness}
\end{proposition}

\begin{proof}
  Given a letter $\smallo\in\mcal{O}$ let $\smallo^\Uparrow$ be the
  set of observations in $\smallo$ and their negations, i.e.,
  observations \emph{not} appearing in $\smallo$.
  That is, $\smallo^\Uparrow = \smallo \cup \{\overline{O} ~|~ O\in
  \obs{\varphi} \setminus \smallo\}$.
  Let $\smallo^\Uparrow_\wedge = \bigwedge_{O\in \smallo^\Uparrow} O$
  be the conjunction of all observations in $\smallo^\Uparrow$. 
  Clearly, the Emptiness of a letter $\smallo\in\mcal{O}$ is
  equivalent to the satisfiability of $\smallo^\Uparrow_\wedge$.
  Thus, we can restrict our attention to the satisfaction of an
  observation.
  Given an observation $O$ let $atom(O)$ denote the set of subformulas
  of $O$ of the form $\exis{O'}$ and $\forall{O'}$.
  
  We show that satisfaction of $O$ can be solved in NP as follows:
  \begin{itemize}
  \item
    select a subset $S$ of $atom(O)$;
  \item
    select an assignment to $\sdat$, a channel $\chan$ and an agent
    $k$;
  \item
    for each $\exis{O'}\in S$ guess one assignment to $\scv$.
  \end{itemize}
  Verify that the choice of $S$, the assignment to $\sdat$, the
  channel $\chan$ and the agent $k$ satisfy $O$.
  Notice, that the elements of $atom(O)$ are treated as Boolean values
  in this check: $O''\in S$ is evaluted as true and
  $O''\notin S$ is evaluated as false.
  For each $\exis{O'}\in S$ check that the assignment to $\scv$
  guessed for $\exis{O'}$ fulfills two conditions:
  \begin{itemize}
  \item
    the assignment to $\scv$ together with the assignment to $\sdat$,
    the channel $\chan$ and the agent $k$ satisfy $O'$.
  \item
    For every $\all{O''}\in S$ check that the assignment to $\scv$
    together with the assignment to $\sdat$, the channel $\chan$ and
    the agent $k$ satisfy $O''$. 
  \end{itemize}

  Notice that the sum of sizes of the guessed elements is polynomal in the
  size of $O$ and the verification can be completed in polynomial time.

  Hardness in \rulename{np} follows from the hardness of Boolean
  satisfiability.
\end{proof}

The following proposition states that checking if a message $m\in M$
satisfies an observation letter $\smallo\in\mcal{O}$ can be tested in
${\rulename{p}}^{\rulename{np}}$ in $\length{\obs{\varphi}}$.
We consider the case that $\pred$ is represented as a  
Boolean formula over $\scv$.
This is reasonable as when considering a transition $(s,m,s')$, where
$m=\tuple{\chan,\datfun,\id,\pred}$, then $\pred$ can be obtained as
such a formula from $\pfunc{g}{s}{\id}$ by using the values in $s$,
$\chan$, and $\datfun$.

\begin{proposition}[Observation model-checking]
  Consider a letter $\smallo\in\mcal{O}$ and an observation $m\in M$.
  Whether $m\models \smallo$ can be tested in
  ${\rulename{p}}^{\rulename{np}}$ in $\length{\obs{\varphi}}$.
  \label{prop:observation checking}
\end{proposition}

\begin{proof}
  As in the case of propsition~\ref{prop:observation emptiness} give a
  letter $\smallo\in\mcal{O}$ we consider $\smallo_\wedge^\Uparrow$.
  Thus, we restrict our attention to the case of whether $m$ satisfies
  an observation $O$. 
  
  Let $m=\tuple{\chan,\datfun,\id,\pred}$.
  We simplify $O$ by converting every reference to $\chan$, $\sdat$ or
  $\id$ to the constants appearing in $m$.
  It follows that we are left with a Boolean combination of
  $\exis{\cdot}$ and $\all{\cdot}$ subformulas, where only
  variables from $\scv$ appear.

  For a subformula $\exis{O'}$ we can check whether $m\models
  \exis{O'}$ by checking whether $\pred\wedge O'$ is satisfiable.
  For a subformula $\all{O'}$ we can check whether $m\models \all{O'}$
  by checking whether $\pred\rightarrow O'$ is valid.
  Both checks can be accomplished by an \rulename{np} oracle.

  The problem is \rulename{np}-hard in $\size{\scv}$ as $m\models
  \exis{\true}$ holds iff $\pred$ is satisfiable.
  The problem is co-\rulename{np}-hard in $\size{\scv}$ as $m\models
  \all{\false}$ iff $\pred$ is unsatisfiable.
  We do not know whether the problem is
  ${\rulename{p}}^{\rulename{np}}$-complete. 
\end{proof}

We note that in the case that $m$ is represented as a set of
assignments to $\scv$, we can modify the Boolean value problem
\cite{lynch} to show that $m\models\smallo$ can be evaluated in
\rulename{logspace}. 

\section{Concluding Remarks}\label{sec:conc}
We introduced a formalism that combines message-passing and shared-memory to facilitate realistic modelling of distributed multi agent systems. 
A system is defined as a set of distributed agents  
 that execute concurrently and only interact by message-passing. Each agent controls its local behaviour as in Reactive Modules~\cite{AH99b,FHNPSV11} while 
interacting externally by message passing as in 
$\pi$-calculus-like formalisms~\cite{pi1,info19}. Thus, we decouple the individual behaviour of an agent from its external interactions to facilitate reasoning about either one separately. We also make it easy to model interaction features of MAS, that may only be tediously hard-coded in existing formalisms.  

We introduced an extension to \ltl, named \ltal, that characterises messages and their targets. This way we may not only be able to reason about the intentions of agents in communication, but also we may explicitly specify their interaction protocols. Finally, we provided a novel automata construction that permits satisfiability and model-checking in space polynomial with respect to the size of the formula and the size of the system. This is a major improvement on the early results in~\cite{rcp} that were in \expspace 
with respect to the number of common variables and \pspace-complete with respect 
to the rest of the input.\smallskip

\noindent
{\bf Related works.}
As mentioned before, formal modelling is highly influenced by traditional
formalisms used for verification, see~\cite{AH99b,FHMV95}.
These formalisms are, however, very abstract in that their models
representations are very close to their mathematical interpretations
(i.e., the underlying transition systems).
Although this  may make it easy to conduct some logical
analysis~\cite{AHK02,CHP10,MMPV14} on models, it does imply that most of the
high-level MAS features may only be hard-coded, and thus leading to
very detailed models that may not be tractable or efficiently
implementable.
This concern has been already recognised and thus more formalisms have
been proposed, e.g., \emph{Interpreted Systems Programming Language}
(ISPL)~\cite{LQR17} and MOCHA~\cite{AHMQRT98} are proposed as
implementation languages of \emph{Interpreted Systems}
(IS)~\cite{FHMV95} and \emph{Reactive Modules} (RM)~\cite{AH99b} 
respectively.
They are still either fully synchronous or shared-memory based and thus do not support
flexible coordination and/or interaction interfaces.
A recent attempt to add dynamicity in this sense has been adopted by
\emph{visibly CGS} (vCGS)~\cite{BBDM19}: an extension of \emph{Concurrent-Game 
Structures} (CGS)~\cite{AHK02} to enable agents to dynamically hide/reveal 
their internal states to selected agents.
However, vCGS relies on an assumption of \cite{stop} which requires that agents 
know the identities of each other.
This, however, only works for closed systems with a fixed number of agents. 

Other attempts to add dynamicity and reconfiguration 
include dynamic I/O automata \cite{DBLP:journals/iandc/AttieL16},
Dynamic reactive modules of Alur and Grosu \cite{AG04},
Dynamic reactive modules of Fisher et
al.~\cite{FHNPSV11}, and open MAS~\cite{KLPP19}.
However, their main interest was in supporting dynamic creation of agents.
Thus, the reconfiguration of communication was not their main
interest. 
While \rcp may be easily extended to support dynamic creation of
agents, none of these formalisms may easily be used
to control the targets of communication and dissemination of
information.

As for logics we differ from traditional languages like \ltl and \ctl 
in that our formula may refer to messages and their constraints.
This is, however, different from the atomic labels of {\pdl}~\cite{FL79} and 
{modal $\mu$-calculus}~\cite{Koz83} in that \ltal mounts complex and structured 
observations on which designers may 
predicate on.
Thus the interpretation of a formula includes information about the causes of 
variable assignments and the interaction protocols among agents.
Such extra information may prove useful in developing compositional 
verification techniques.

\smallskip

\noindent
{\bf Future works.}  
We plan to provide tool support for \rcp, but with a more 
user-friendly syntax.
We would like to provide a light-weight programming-language-like
syntax to further simplify modelling.

We want to exploit the interaction mechanisms in \rcp and the extra
information in \ltal formulas to conduct verification
compositionally.
As mentioned, we believe that relating to sender intentions will
facilitate that. 

We intend to study the relation with respect to temporal epistemic logic 
\cite{HV89}.
Although we do not provide explicit knowledge operators, the combination of 
data exchange, receivers selection, and enabling/disabling of synchronisation based on the evolving states of the different agents,
allow them to dynamically deduce information about each other.
Furthermore we want to study \rcp under dynamic creation of agents while 
reconfiguring communication. Thanks to the new compositional semantics in terms of CTS, the dynamic creation of agents can now be easily linked to the execution of some blocking transitions. To give an intuition about a (possible) extension consider the following semantic rule:

\begin{equation}\label{eq:rep}
{A_1 \rTo{(\upsilon,?,c)}A'_1\| A_1}{}
\end{equation}

Here, we use $A\rTo{m}A'$ to denote that in state $s$ agent $A$ may
receive a message $m$ and evolves to $A'$ with state $s'$, i.e.,
$(s,m,s')\in R_A$. Clearly, the semantic rule~\ref{eq:rep} indicates
that the agent replicates itself once a multicast message on $c$ is
received as a side-effect of interaction. Thus, if we compose $A_1$
with some other agent, say $A_2$, such that
$A_2\rTo{(\upsilon,!,c)}A'_2$, the following transition is derivable
by the semantics of parallel composition in Def.~\ref{def:par}: 

\begin{equation}\label{eq:rep1}
\infer{A_1\| A_2 \rTo{(\upsilon,!,c)}A'_1\| A_1\| A_2}{\text{if}\quad A_1\rTo{(\upsilon,?,c)}A'_1\| A_1\qquad\text{and}\qquad A_2\rTo{(\upsilon,!,c)}A'_2}
\end{equation}

Namely, a new replica of $A_1$ is dynamically created when agents exchanged a specific message. Agent $A_1$ can be thought of as a server that spawn a new thread to handle concurrent requests from clients. 
 
Finally, we want to target
the distributed synthesis problem~\cite{FS05}.
Several fragments of the problem have been proven to be decidable, e.g.,
when the information of agents is arranged 
hierarchically~\cite{BMMRV17}, the number of agents is limited~\cite{GPW18}, or 
the actions are made public~\cite{BLMR17a}.
We conjecture that the ability to disseminate information and reason about it might prove 
useful in this setting.

%
%



\end{document}